\documentclass[a4paper,UKenglish,cleveref, thm-restate,autoref]{lipics-v2021}

\hideLIPIcs  

\usepackage{tikz}
\usepackage{mathtools}

\newcommand{\CE}{\textnormal{\textsc{Cluster Editing}}}

\newcommand{\CEwFD}{\textnormal{\textsc{Cluster Editing with Few Deletions}}}

\newcommand{\HFEwFD}{\textnormal{\textsc{$H$-Free Editing with Few Deletions}}}
\newcommand{\MCC}{\textnormal{\textsc{Multicolored Clique}}}
\newcommand{\fvs}{\operatorname{fvs}}
\newcommand{\fes}{\operatorname{fes}}
\newcommand{\pw}{\operatorname{pw}}
\newcommand{\proofparagraph}[1]{\smallskip\noindent\textbf{#1.}}
\newcommand{\mods}{\operatorname{mod}}
\newcommand{\miss}{\operatorname{miss}}
\newcommand{\yes}{\textup{\texttt{yes}}}
\newcommand{\no}{\textup{\texttt{no}}}
\newcommand{\lb}{\operatorname{lb}}

\newtheorem{branchingrule}[theorem]{Branching Rule}
\newtheorem{reductionrule}[theorem]{Reduction Rule}
\declaretheorem[style=definition,name=Construction,qed=$\diamond$]{construction}

\Crefname{observation}{Observation}{Observations}
\Crefname{branchingrule}{Branching Rule}{Branching Rules}
\Crefname{reductionrule}{Reduction Rule}{Reduction Rules}
\Crefname{claim}{Claim}{Claims}

\title{Parameterizing Cluster Editing\\by the Number of Deletions}

\titlerunning{Parameterizing Cluster Editing by the Number of Deletions} 

\author{Vincent Froese}{Technische Universität Berlin, Algorithmics and Computational Complexity Research Group, Berlin, Germany}{vincent.froese@tu-berlin.de}{https://orcid.org/0000-0002-8499-0130}{}

\author{Leon Kellerhals}{Technische Universität Clausthal, Clausthal-Zellerfeld, Germany}{leon.kellerhals@tu-clausthal.de}{https://orcid.org/0000-0001-6565-3983}{}

\author{Frank Sommer}{Institute of Computer Science, Friedrich Schiller University Jena, Germany}{frank.sommer@uni-jena.de}{https://orcid.org/0000-0003-4034-525X}{}

\author{Manuel Sorge}{TU Wien, Institute of Computer Science, Algorithms and Complexity Group, Vienna, Austria}{manuel.sorge@ac.tuwien.ac.at}{https://orcid.org/0000-0001-7394-3147}{}

\author{Ondřej {Suchý}}{Faculty of Information Technology, Czech Technical University in Prague, Prague, Czech Republic}{ondrej.suchy@fit.cvut.cz}{https://orcid.org/0000-0002-7236-8336}{}

\authorrunning{Vincent Froese, Leon Kellerhals, Frank Sommer, Manuel Sorge, Ondřej {Suchý}} 

\Copyright{Vincent Froese, Leon Kellerhals, Frank Sommer, Manuel Sorge, Ondřej {Suchý}} 

\ccsdesc[500]{Theory of computation~Parameterized complexity and exact algorithms}
\ccsdesc[500]{Theory of computation~Graph algorithms analysis}

\keywords{Cluster Editing, parameterized complexity, solution-based parameters, above-guarantee parameterization, kernelization, W[1]-hardness}

\category{} 

\relatedversion{} 

\acknowledgements{This work was initiated at a research retreat around Rolf Fest held in Jena, Germany, in July 2026. We thank Falk Hüffner for interesting discussions.}

\nolinenumbers 

\EventEditors{John Q. Open and Joan R. Access}
\EventNoEds{2}
\EventLongTitle{42nd Conference on Very Important Topics (CVIT 2016)}
\EventShortTitle{CVIT 2016}
\EventAcronym{CVIT}
\EventYear{2016}
\EventDate{December 24--27, 2016}
\EventLocation{Little Whinging, United Kingdom}
\EventLogo{}
\SeriesVolume{42}
\ArticleNo{23}

\usepackage{booktabs}
\usepackage[disable]{todonotes}

\usepackage{etoolbox}
\def\ShortVersion{1} 
\ifdefined\ShortVersion
\newcommand{\sv}[1]{#1}
\newcommand{\lv}[1]{}
\newcommand{\appendixText}{}
\newcommand{\toappendix}[1]{\gappto{\appendixText}{{#1}}}
\else
\newcommand{\sv}[1]{}
\newcommand{\lv}[1]{#1}
\newcommand{\appendixText}{}
\newcommand{\toappendix}[1]{#1}
\fi

\newcommand{\appmark}{$\star$}

\begin{document}

\maketitle

\begin{abstract}
In \CE{}, we are given a graph~$G$ and an integer~$k$ and ask whether at most~$k$ edge insertions and deletions turn~$G$ into a disjoint union of cliques.
Komusiewicz and Niedermeier asked in 2012 at MFCS whether \CE{} is fixed-parameter tractable with respect to the number~$\ell$ of edge deletions, a parameter based on the solution rather than on the structure of the input graph.
We answer this question negatively: \CE{} with at most~$k$ edits, at most~$\ell$ of which are deletions, is W[1]-hard with respect to~$\ell$, even when combined with the pathwidth and the feedback edge set number of the input graph.
The hardness extends to editing towards $H$-free graphs for every graph~$H$ that is a clique minus one edge on at least four vertices.
On the positive side, we show that the problem becomes fixed-parameter tractable when $\ell$ is combined with the excess~$\mu$ of~$k$ over the lower bound given by a packing of modification-disjoint induced paths on three vertices, a lower bound that state-of-the-art solvers use.
Each of the two parameters alone leads to intractability, but together they admit an $O(4^{\mu+\ell}\cdot(n+m+\mu))$-time algorithm and a problem kernel with $O((\ell+\mu)^2)$ vertices.
Both results extend to packings of stars and the algorithm extends to \textsc{$H$-free editing}.
\end{abstract}


\clearpage

\section{Introduction}
\looseness=-1
In \CE\ we are given a graph $G$ and an integer $k$ and want to know whether at most $k$ edges can be toggled (deleted if present, added if not) such that $G$ becomes a \emph{cluster graph}, a disjoint union of cliques (equivalently, a $P_3$-free graph).
It is one of the proving grounds for parameterized algorithmics, serving as touchstone for kernelization~\cite{GrammGHN05,ProttiSS09,ChenM12}, subexponential algorithms and matching running time lower bounds~\cite{FominKPPV14,KomusiewiczU12}, fixed-parameter enumeration~\cite{Damaschke10}, structural and other alternative parameterizations~\cite{
  KomusiewiczU12,A17}, above-guarantee parameterizations~\cite{BevernFK18,LiPS24}, and algorithm engineering~\cite{BockerBK11,HartungH15,BlasiusFGHHSWW22}, including the PACE challenge~2021~\cite{KellerhalsKNZ21}.
It was subject to a race improving the base of the exponential running-time part from $2.27$~\cite{GrammGHN05} 
via a series of improvements \cite{GrammGHN04,BockerBBT09,BockerD11}
to $1.62$~\cite{Bocker12}, and to a race decreasing the kernel size from $O(k^2)$ vertices~\cite{GrammGHN05} via $24k$~\cite{FellowsLRS07} and $4k$~\cite{Guo09} to $2k$ vertices~\cite{ChenM12,CaoC12}.

About 15 years ago, another race started for many parameterized problems; the race to prove tractability results for smaller and smaller parameters~\cite{Niedermeier10,FellowsJR13,KomusiewiczN12} (see also ref.s~\cite{Sorge26,HOPS,ISGCI}).
The most salient direction that researchers focused on here was to explore how graph problems behave under different graph-structural parameterizations.  
This has resulted in an immense range of papers and it is by now fairly well understood which structural graph parameters yield tractability for which classical problems; see, e.g., algorithmic meta-theorems~\cite{Courcelle90,CourcelleMR00,BonnetKTW22}, the studies introducing various tructural parameters~\cite{FellowsLMRS08,Ganian15,Lampis12,GajarskyLO13,DouchaK12}
, and systematic studies of individual problems~\cite{HartungKN15,MajumdarR18}.
The initial ideas were much broader, however.
For instance, Komusiewicz and Niedermeier~\cite{KomusiewiczN12} proposed studying solution-based parameters: parameterizing by features of the solution, including above- or below-guarantee parameterizations, but also other structural features of the solution.
While a number of case studies exist~\cite{ItoI09,KomusiewiczHMN09,KomusiewiczU12,LokshtanovM13,FominKPPV14,BevernFSS15,HuffnerKS15,ChechikJPP17,A17,FominGKK17,BevernFMMSS18,GolovachHLM20}, solution-based parameters have been studied far less systematically and are less well understood.
We contribute to closing this gap by studying a natural solution-based parameter for \CE.

In particular, Komusiewicz and Niedermeier asked whether \CE\ is fixed-parameter tractable with respect to the number $\ell$ of allowed edge deletions~\cite[p.~26]{KomusiewiczN12}:
Clearly, $\ell \leq k$ and the case $\ell = 0$ is trivial.
More generally, there is a trivial $O(n^{2\ell + 2})$-time algorithm.
In \cref{sec:exp} we show that on a large fraction of PACE Challenge 2021 instances, the median value of $\ell$ is $0.73$ times the optimal number of edits, which would make it an interesting replacement for $k$ in FPT algorithms.
Bounding the modifications \emph{locally} has been studied: Komusiewicz and Uhlmann~\cite{KomusiewiczU12} bound the number of edits incident to each vertex, and Abu-Khzam~\cite{A17} separately bounds the numbers of additions and deletions incident to each vertex, showing NP-hardness even when each vertex is incident to at most one deletion and two additions.
The global number~$\ell$ of deletions, however, has surprisingly not been studied so far; we aim to close this gap.
The bad news up front is that \CE\ parameterized by~$\ell$ is W[1]-hard.

\looseness=-1
Towards the good news, a second motivation for our research comes from the unreasonable 
effectiveness of practical implementations for \CE:
The most effective algorithms are based on fixed-parameter algorithms for $k$~\cite{KellerhalsKNZ21,BlasiusFGHHSWW22}.
But this parameter is relatively large in practical data (see \cref{sec:exp}), yet the algorithms are quite efficient.
What explains this?
One direction that researchers have explored is that the lower bounds employed by the implemented algorithms~\cite{GottesburenHSSW20,BlasiusFGHHSWW22} are quite close to the optimal solution (see also \cref{sec:exp}); in particular they use packings of modification-disjoint $P_3$s or stars, where modification disjointness means that each two packed $P_3$s do not overlap in two vertices.
However, parameterizing \CE\ by the number $\mu$ of edits beyond the lower bound given by a modification-disjoint packing of $P_3$s results in a para-NP-hard problem~\cite{LiPS24}. 
Hence, this is no theoretical explanation of the observed efficiency.
While parameterizing above a \emph{vertex-disjoint} packing of $P_3$s yields fixed-parameter tractability~\cite{BevernFK18}, this parameter is still often large in practice (\cref{sec:exp}).
Towards explaining this discrepancy, we in particular show that \CE\ becomes fixed-parameter tractable when parameterizing above a modification-disjoint $P_3$-packing (which is practically close to zero) together with the number $\ell$ of deletions (which is often significantly smaller than $k$, see \cref{sec:exp}).
It also allows for a polynomial-size problem kernel wrt. $\ell$ and $\mu$ combined and both these results transfer to the case where we replace $\mu$ by the excess $\sigma$ over a modification-disjoint star packing.

\looseness=-1
Towards a more general understanding, we then broaden our attention to \textsc{$H$-Free Editing}, where for some fixed graph $H$ we want to toggle at most $k$ edges in an input graph $G$ to ensure that the resulting graph is $H$-free, that is, it has no copy of $H$ as an induced subgraph.
Similar to \CE, \textsc{$H$-Free Editing} is NP-hard~\cite{AravindSS17} and FPT wrt.~$k$~\cite{Cai96}.
A well-known special case is \textsc{Diamond-free Editing}, which corresponds to $H$ being $K_4$ missing one edge, where even polynomial-size problem kernels are known~\cite{CRSY22}; their existence is settled for some other special cases of $H$ as well~\cite{DBLP:journals/jcss/MarxS22}.
When parameterizing \textsc{$H$-Free Editing} with the number of deletions $\ell$, the picture changes:
Aravind, Sandeep, and Sivadasan~\cite{AravindSS17} showed that adding edges to make a graph $H$-free is NP-hard whenever $H$ misses at least two edges.
This leaves \textsc{$H$-Free Editing} NP-hard for $\ell = 0$ in these cases.
Otherwise, if $H$ misses at most one edge, then \HFEwFD\ is trivially in XP by guessing the deletions and then adding the missing edge for each induced copy of $H$.
As we show, it is also W[1]-hard with respect to $\ell$.
On the positive side, we show that the above fixed-parameter algorithms for \CE\ with respect to $\ell$ combined with $\mu$ transfer to \textsc{$H$-Free Editing} when replacing the underlying packing of $\mu$ with a packing of copies of $H$.
We leave open whether polynomial-size problem kernels exist for this parameterization. 



\subparagraph*{Organization and results}
After formally defining \CEwFD{} and \HFEwFD{} and collecting notation in \cref{sec:preliminaries}, we prove our hardness results in \cref{sec:hardness}.
\Cref{thm:ce-hardness} shows that \CEwFD{} is W[1]-hard parameterized by~$\ell$, even combined with the pathwidth and the feedback edge set number of the input graph. 
This answers the question of Komusiewicz and Niedermeier~\cite[p.~26]{KomusiewiczN12} negatively.
The reduction is from \MCC{} and adapts a reduction of Enciso et al.~\cite{EncisoFGKRS09} for equitable connected partitions.
\Cref{thm:hfree-hardness} extends the hardness to \HFEwFD{} for every $H=K_\eta-e$ with $\eta\geq4$, with the feedback vertex set number in place of the feedback edge set number.
In \cref{sec:above-guarantees}, we combine $\ell$ with the excess~$\mu$ above a modification-disjoint $P_3$-packing.
\Cref{thm:fpt-mu-ell} solves \CEwFD{} in $O(4^{\mu+\ell}\cdot(n+m+\mu))$ time, \cref{thm:hfree-fpt-mu-ell} transfers fixed-parameter tractability to \HFEwFD{}, with the excess over a modification-disjoint packing of induced copies of~$H$ in place of~$\mu$, and \cref{thm:fpt-sigma-ell} replaces the $P_3$-packings in \cref{thm:fpt-mu-ell} by packings of stars with corresponding excess~$\sigma$.
We then give a kernel for \CEwFD\ with $O((\ell+\mu)^2)$ vertices and $O((\ell+\mu)^3)$ edges (\cref{thm:kernel-mu-ell}) and a polynomial kernel for $\ell+\sigma$ (\cref{thm:kernel-sigma-ell}).
\Cref{sec:discussion} concludes with open questions.
Proofs of results marked with~\appmark{} are deferred to the appendix, which also contains our experiments (\cref{sec:exp}).

\section{Preliminaries}\label{sec:preliminaries}

A \emph{cluster graph} is a graph whose connected components are cliques. This is equivalent to not containing any induced $P_3$ (a path on three vertices).
For an undirected graph $G=(V,E)$ and a set $F\subseteq\binom{V}{2}$, the graph
$G\triangle F:=(V,E\triangle F)$, where~$\triangle$ is the symmetric difference, is obtained by toggling every pair in $F$.
The pairs in $F\cap E$ are the \emph{edge deletions}, and the pairs in
$F\setminus E$ are the \emph{edge additions}.
Our study deals mainly with the following problem.
\begin{center}
\begin{tabularx}{\textwidth}{@{}lX@{}}
\multicolumn{2}{l}{\CEwFD}\smallskip\\
\textbf{Input:} & An undirected graph $G=(V,E)$ and non-negative integers $k$ and
$\ell$.\\
\textbf{Question:} & Is there a set $F\subseteq\binom{V}{2}$ such that
$|F|\leq k$, $|F\cap E|\leq\ell$, and $G\triangle F$ is a cluster graph?
\end{tabularx}
\end{center}
We also consider the following generalization.
Let $H$ be a graph.
A graph $G$ is $H$-free if it does not contain an induced subgraph isomorphic to $H$.
For example, the $P_3$-free graphs are exactly the cluster graphs.
\begin{center}
\begin{tabularx}{\textwidth}{@{}lX@{}}
\multicolumn{2}{l}{\textsc{$H$-Free Editing with Few Deletions}}\smallskip\\
\textbf{Input:} & An undirected graph $G=(V,E)$ and non-negative integers $k$ and
$\ell$.\\
\textbf{Question:} & Is there a set $F\subseteq\binom{V}{2}$ such that
$|F|\leq k$, $|F\cap E|\leq\ell$, and $G\triangle F$ is $H$-free?
\end{tabularx}
\end{center}
Our algorithms use lower bounds given by packings of induced subgraphs.
For an induced subgraph~$Q$ of~$G$, we call $\mods(Q)\coloneqq\binom{V(Q)}{2}$ the \emph{modification set} of~$Q$: toggling a pair in $\mods(Q)$ changes $G[V(Q)]$, whereas toggling a pair outside $\mods(Q)$ does not.
A collection of induced subgraphs of~$G$ is \emph{modification disjoint} if their modification sets are pairwise disjoint, that is, if any two of them share at most one vertex.
If $Q$ is isomorphic to $K_h-e$ for some $h\geq3$, for instance if $Q$ is a~$P_3$, then we call the unique pair of nonadjacent vertices of~$Q$ its \emph{missing pair} $\miss(Q)$.

For a graph $G$, let $\pw(G)$ denote its pathwidth, let $\fvs(G)$ denote the
minimum size of a vertex set whose removal makes $G$ a forest, and let
$\fes(G)$ denote the minimum size of an edge set whose removal makes $G$ a
forest.  We write $[p]:=\{1,\ldots,p\}$ for a positive integer $p$.


\looseness=-1\subparagraph*{Parameterized complexity.}
We use standard notions of parameterized complexity~\cite{CyganFKLMPPS15,DowneyF13,Niedermeier06}.
A \emph{decision problem} is a set~$\Pi$ whose elements are called \emph{instances}, together with a subset of~$\Pi$ whose elements are called the \emph{yes-instances}. 
The remaining instances are the \emph{no-instances}.
A \emph{parameterized problem} is a pair $(\Pi,\kappa)$ of a decision problem~$\Pi$ and a computable function $\kappa\colon\Pi\to\mathbb{N}$, the \emph{parameterization}; for an instance~$I$, the number $\kappa(I)$ is its \emph{parameter}.
The problem is \emph{fixed-parameter tractable} (FPT) if there is an \emph{FPT-algorithm}, that decides every instance~$I$ in $f(\kappa(I))\cdot|I|^{O(1)}$ time for some computable function~$f$.
Showing that a parameterized problem is \emph{W[1]-hard} shows that it is presumably not fixed-parameter tractable, since otherwise FPT${}={}$W[1], which is considered unlikely.
A \emph{problem kernel} is a polynomial-time algorithm that maps every instance~$I$ to an equivalent instance~$I'$, that is, $I$ is a yes-instance if and only if $I'$ is, such that $|I'|\leq h(\kappa(I))$ for some computable function~$h$, the \emph{size} of the kernel.
If $h$ is a polynomial, then the kernel is a \emph{polynomial kernel}.

One common technique for designing FPT-algorithms 
is \emph{branching}: 
A \emph{branching rule} for $\Pi$ is a partial map $R$ from $\Pi$ to the set of finite subsets of $\Pi$, and the instances on which $R$ is defined are the instances that $R$ \emph{applies to}.
The rule $R$ is \emph{safe} if every instance $I$ that it applies to satisfies:
$I$ is a yes-instance of $\Pi$ if and only if $R(I)$ contains a yes-instance of $\Pi$.

\section{Hardness}\label{sec:hardness}
\toappendix{\sv{\section{Missing Details for the Hardness Results}\label{app:hardness}}}

We resolve the open question of Komusiewicz and Niedermeier~\cite{KomusiewiczN12} in the negative by showing that \CEwFD{} is W[1]-hard with respect to~$\ell$ (\cref{sec:ce-hardness}).
Afterwards, we generalize the result to \HFEwFD{} (\cref{sec:hfree-hardness}) for all graphs $H$ that are a clique minus one edge, such as diamonds.
Both reductions share a common core, the \emph{choice graph}, which we introduce first (\cref{sec:choice-graph}).

The hardness is based on the following idea.
With only few deletions allowed, we may be forced to cut an input graph into connected components such that the remaining edit budget suffices to complete them into cliques.
Since completing a sparse graph requires roughly a quadratic number of edge additions,
the overall number of additions is minimized if the connected components after the deletions have pairwise equal size.
Indeed, Enciso et al.~\cite[Thm.~1]{EncisoFGKRS09} showed that partitioning a graph into $r$ connected parts whose sizes differ by at most one is W[1]-hard parameterized by~$r$.
We closely follow their reduction from \MCC, which encodes vertex and edge choices by the positions at which specific paths are cut.
Recall that an instance of \MCC{} consists of a graph $J$ together with a partition $U_1, \dots, U_p$ of its vertex set~$U$ into \emph{color classes}, and the parameter $p$; we refer to the index~$i$ of $U_i$ as a \emph{color}.
It is a \yes-instance if $J$ contains a clique containing exactly one vertex of every color class.
We assume without loss of generality that each color class contains $n_0$ vertices, say $U_i=\{u^i_1,\ldots,u^i_{n_0}\}$, and that there are at least $p \ge 3$ color classes.

\subsection{Choice Gadgets and the Choice Graph}\label{sec:choice-graph}
\toappendix{\sv{\subsection{Choice Gadgets and the Choice Graph}\label{app:choice-graph}}}

Both reductions are based on the choice graph which allows us to select the vertices and edges of the clique in $J$ in terms of so-called \emph{choices}.
The choice graph is built from so-called \emph{choice gadgets} which differ between the two reductions, but share the following structure.
Let $b$ be a positive integer and let $A=\{a_1, \ldots, a_T\}\subseteq[b]$ with $a_1 < \cdots < a_T$.
An \emph{$(A,b)$-choice gadget} from a vertex~$x$ to a vertex~$y$ is a graph~$Q$ with $x,y\in V(Q)$, the \emph{endpoints} of~$Q$,
together with $T$ \emph{options} $o_1,\ldots,o_T$, each of which is a vertex or an edge of~$Q$;
the option~$o_t$ has \emph{value}~$a_t$.
The vertices in $V(Q)\setminus\{x,y\}$ are the \emph{inner vertices} of~$Q$.

For integers $\lambda\ge1$ and $c\ge0$, we say that $Q$ has the \emph{$(\lambda,c)$-option property} if, for every $t\in[T]$,
removing~$o_t$ from~$Q$ leaves a connected component containing~$x$ with exactly $\lambda a_t+c$ inner vertices
and a different connected component containing~$y$ with exactly $\lambda(b-a_t)+c$ inner vertices.
That is, choosing an option of value~$a$ splits the gadget such that $x$ receives a share that is proportional to~$a$ and $y$ a share that is proportional to~$b-a$.

\begin{construction}[Choice graph]
  \label{constr:choice-graph}
  We are given a choice gadget construction and an instance of \MCC{} with above naming and assumptions.
  Let $m_0 \coloneqq |E(J)|$ and fix an arbitrary ordering of the edges such that $E(J) = \{e_1, \dots, e_{m_0}\}$.
  We call $\gamma$ the \emph{identifier} of the edge~$e_\gamma$.
  Set $z_0 \coloneqq 2m_0 + 10$ and, for distinct $i,j \in [p]$, define
  \begin{align*}
    A_i^0 &:= \{\theta z_0 \mid \theta\in[n_0]\},\\
    A_i^j &:= \{\theta z_0+\gamma \mid \theta\in[n_0],\ u\in U_j,\
    \{u^i_\theta,u\}\in E(J) \text{ such that edge } \{u^i_\theta,u\} \text{ has identifier } \gamma\}.
  \end{align*}
  For each ordered color pair $(i,j)$ with $i\neq j$, create two \emph{anchor roots} $N_i^j$ and $P_i^j$.
  For $i\in[p]$, let $\operatorname{succ}_i$ be the cyclic permutation of $[p]\setminus\{i\}$
  that maps every color to the next larger one, and the largest one to the smallest one.
  For every ordered color pair~$(i,j)$, add the following three new choice gadgets (see \cref{fig:choice-graph}):
  \begin{enumerate}
  \item a \emph{local gadget}, which is an $(A_i^j,n_0 z_0+m_0)$-choice gadget
        from $N_i^j$ to $P_i^j$;
  \item a \emph{successor gadget}, which is an $(A_i^0,n_0 z_0)$-choice gadget
        from $P_i^j$ to $N_i^{\operatorname{succ}_i(j)}$; and
  \item a \emph{mirror gadget}, which is an $([m_0],m_0)$-choice gadget from $P_i^j$
        to $N_j^i$.
  \end{enumerate}
  Let $\mathcal Q$ be the set of all these choice gadgets.
  The resulting graph is the \emph{choice graph}.
\end{construction}
\begin{figure}[t]
\centering
\begin{tikzpicture}[
  anchorroot/.style={draw, circle, fill=white, minimum size=7mm, inner sep=1pt,
                     font=\scriptsize},
  local/.style={black, line width=.8pt},
  successor/.style={blue!70!black, densely dotted, line width=1pt},
  mirror/.style={red!70!black, dashed, line width=.9pt},
  label/.style={font=\scriptsize}
]
  \coordinate (N12) at (-4.25,.65);
  \coordinate (P12) at (-2.95,.65);
  \coordinate (N13) at (-2.95,-.65);
  \coordinate (P13) at (-4.25,-.65);
  \coordinate (N21) at (-.65,.65);
  \coordinate (P21) at (.65,.65);
  \coordinate (N23) at (.65,-.65);
  \coordinate (P23) at (-.65,-.65);
  \coordinate (N31) at (2.95,.65);
  \coordinate (P31) at (4.25,.65);
  \coordinate (N32) at (4.25,-.65);
  \coordinate (P32) at (2.95,-.65);

  \draw[local] (N12) -- (P12) (N13) -- (P13);
  \draw[local] (N21) -- (P21) (N23) -- (P23);
  \draw[local] (N31) -- (P31) (N32) -- (P32);
  \draw[successor] (P12) -- (N13) (P13) -- (N12);
  \draw[successor] (P21) -- (N23) (P23) -- (N21);
  \draw[successor] (P31) -- (N32) (P32) -- (N31);
  \draw[mirror] (P12) -- (N21);
  \draw[mirror] (P21) .. controls (.2,1.45) and (-3.8,1.45) .. (N12);
  \draw[mirror] (P32) -- (N23);
  \draw[mirror] (P23) .. controls (-.2,-1.45) and (3.8,-1.45) .. (N32);
  \draw[mirror] (P31) .. controls (4,2) and (-2.7,2) .. (N13);
  \draw[mirror] (P13) .. controls (-4,-2) and (2.7,-2) .. (N31);

  \node[anchorroot] at (N12) {$N_1^2$};
  \node[anchorroot] at (P12) {$P_1^2$};
  \node[anchorroot] at (N13) {$N_1^3$};
  \node[anchorroot] at (P13) {$P_1^3$};
  \node[anchorroot] at (N21) {$N_2^1$};
  \node[anchorroot] at (P21) {$P_2^1$};
  \node[anchorroot] at (N23) {$N_2^3$};
  \node[anchorroot] at (P23) {$P_2^3$};
  \node[anchorroot] at (N31) {$N_3^1$};
  \node[anchorroot] at (P31) {$P_3^1$};
  \node[anchorroot] at (N32) {$N_3^2$};
  \node[anchorroot] at (P32) {$P_3^2$};
  \node[label] at (-3.6,0) {color 1};
  \node[label] at (0,0) {color 2};
  \node[label] at (3.6,0) {color 3};

  \draw[local] (-4.25,-2.0) -- (-3.55,-2.0)
    node[label, right] {local gadget};
  \draw[successor] (-.75,-2.0) -- (-.05,-2.0)
    node[label, right] {successor gadget};
  \draw[mirror] (2.55,-2.0) -- (3.25,-2.0)
    node[label, right] {mirror gadget};

  \begin{scope}[shift={(6.4,0)},
      spinev/.style={draw, circle, fill=white, minimum size=4mm, inner sep=0pt, font=\scriptsize},
      leafv/.style={draw, circle, fill=black, minimum size=1.6mm, inner sep=0pt}]
    \coordinate (v1) at (0,1.5);  \coordinate (v2) at (0,.5);
    \coordinate (v3) at (0,-.5);  \coordinate (v4) at (0,-1.5);
    \draw[local] (v1) -- node[label, left] {$o_1$: $2$} (v2)
                      -- node[label, left] {$o_2$: $3$} (v3)
                      -- node[label, left] {$o_3$: $5$} (v4);
    \foreach \v/\dy in {v1/.2, v1/-.2, v3/0, v4/0}
      \draw (\v) -- ++(.75,\dy) node[leafv] {};
    \node[spinev] at (v1) {$x$};
    \node[spinev] at (v2) {};
    \node[spinev] at (v3) {};
    \node[spinev] at (v4) {$y$};
  \end{scope}
\end{tikzpicture}
\caption{Left: the choice graph for $p=3$.  Each choice gadget is represented by
one edge.  For each color, the local and successor gadgets form a
four-cycle; the mirror gadgets join different color cycles.
Right: the $(\{2,3,5\},6)$-choice gadget from~$x$ to~$y$ of \cref{constr:ce-choice-gadget}; deleting its option~$o_t$ (labeled with its value) leaves exactly that many of its six inner vertices in the component of~$x$.}
\label{fig:choice-graph}
\end{figure}
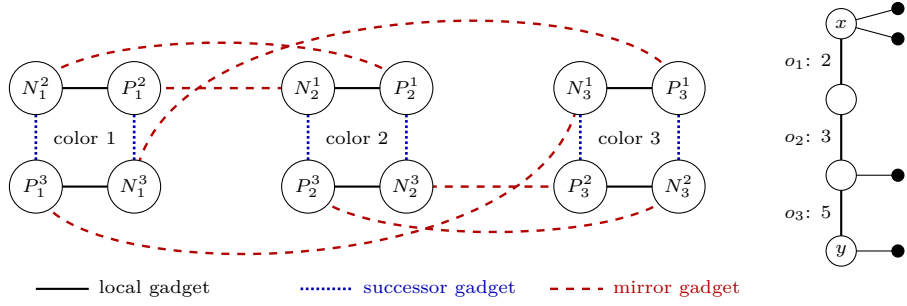

The choice graph has
$r \coloneqq 2p(p-1)$ anchor roots and $q \coloneqq |\mathcal Q|=3p(p-1)$ choice gadgets,
which pairwise intersect at most in a common endpoint.
Every anchor root is an endpoint of exactly three choice gadgets:
$P_i^j$ of the three choice gadgets for~$(i,j)$, and
$N_i^j$ of the local gadget for~$(i,j)$, the successor gadget for~$(i,\operatorname{succ}_i^{-1}(j))$, and the mirror gadget for~$(j,i)$.
If all choice gadgets are connected, then so is the choice graph:
the local and successor gadgets of color~$i$ form a cycle through all color-$i$ anchor roots; the mirror gadgets connect these cycles.

A \emph{selection} chooses one option~$o_Q$ of every choice gadget $Q\in\mathcal Q$.
If $Q$ is an $(A_Q,b_Q)$-choice gadget from~$x$ to~$y$ and $o_Q$ has value~$a$,
then $Q$ \emph{contributes} $a$ to~$x$ and $b_Q-a$ to~$y$.
The selection is \emph{balanced} if every anchor root receives contributions summing to exactly
$  d \coloneqq n_0 z_0+m_0$
from its three incident choice gadgets.
Note that these notions only depend on the parameters and the options of the choice gadgets, but not on their structure.

\begin{restatable}[\appmark]{lemma}{balancedSelection}\label{lem:balanced-selection}
  For every choice gadget construction, the choice graph returned from \cref{constr:choice-graph}
  has a balanced selection
  if and only if
  the given instance of \MCC{} is a \yes-instance.
\end{restatable}

\looseness=-1
\sv{Reading~$d$ in base~$z_0$ explains the construction:
its $z_0$-digit~$n_0$ represents a vertex index, and its unit digit~$m_0$ represents an edge identifier.
Choosing the option of value $\theta z_0+\gamma$ in the local gadget for~$(i,j)$ thus chooses the vertex~$u^i_\theta$ together with an incident edge into~$U_j$ with identifier~$\gamma$.
The successor gadgets pass on vertex indices only, forcing all local gadgets of color~$i$ to choose the same vertex;
the mirror gadgets pass on edge identifiers only, forcing the local gadgets for~$(i,j)$ and~$(j,i)$ to choose the same edge.
\sv{We give the details in \cref{app:choice-graph}.}}

\toappendix{
\sv{In this section, we prove \cref{lem:balanced-selection}.}
We first describe in more detail how the choice graph works.
Since $m_0<z_0$, each element of $A_i^j$ can be written as $\theta z_0+\gamma$ with
$1\leq\gamma<z_0$ in exactly one way: the factor $\theta$ identifies the color-$i$
vertex $u^i_\theta$, and the remainder $\gamma$ is the identifier of an edge with
endpoints $u^i_\theta$ and a vertex of $U_j$.  A single number in $A_i^j$ thus
encodes a color-$i$ vertex together with an incident edge into $U_j$, whereas
the numbers in $A_i^0$ carry only a vertex index and those in $[m_0]$ only an
edge identifier.

Reading the demand $d=n_0z_0+m_0$ in base $z_0$ explains the three
kinds of choice gadgets.  Its $z_0$-digit $n_0$ is a budget for a vertex index and its
unit digit $m_0$ is a budget for an edge identifier.  Choosing an option of the local
gadget for $(i,j)$ contributes some $\theta z_0+\gamma$ to $N_i^j$ and the complementary
$(n_0-\theta)z_0+(m_0-\gamma)$ to $P_i^j$; it therefore encodes both a color-$i$
vertex $u^i_\theta$ and an edge from $u^i_\theta$ into $U_j$.
Choosing an option of the successor gadget
for $(i,j)$ contributes some $wz_0$ to $P_i^j$ and $(n_0-w)z_0$ to
$N_i^{\operatorname{succ}_i(j)}$, so it encodes a vertex index and nothing
else; choosing an option of the mirror gadget for $(i,j)$ contributes some $g$ to $P_i^j$ and
$m_0-g$ to $N_j^i$, so it encodes an edge identifier and nothing else.
Since $z_0>m_0$, an identifier never reaches the $z_0$-digit; this is what
lets us read both a vertex index and an edge identifier off a single
equation.

How the choice gadgets are wired then turns these local decisions into the two
consistency requirements of a multicolored clique.  The local and successor
gadgets of color $i$ alternate along the cycle through all $2(p-1)$ anchor
roots of color $i$, and a successor gadget passes on a vertex
index only; hence all $p-1$ local gadgets for color $i$ must choose one and the
same vertex of $U_i$.  A mirror gadget links the ordered pairs $(i,j)$ and
$(j,i)$ and passes on an edge identifier only; hence these two pairs must
choose the same edge identifier and therefore the same edge.  That edge then
joins the vertex chosen for color $i$ to the vertex chosen for color $j$, so
the $p$ chosen vertices are pairwise adjacent.
We now make this formal.

\sv{\balancedSelection*}
\begin{proof}
Suppose first that a balanced selection is given.  We introduce the
following quantities for every ordered pair $(i,j)$ with $i\neq j$.  The local
gadget from $N_i^j$ to $P_i^j$ contributes an element of $A_i^j$ to $N_i^j$;
write this element as $\theta_{ij}z_0+\gamma_{ij}$ with $1\leq\gamma_{ij}<z_0$,
which is possible in exactly one way, so that $\gamma_{ij}$ is the identifier
of an edge with endpoints $u^i_{\theta_{ij}}$ and a vertex of $U_j$.
Let the successor gadget leaving $P_i^j$ contribute $w_{ij}z_0$ to
$P_i^j$, and let the mirror gadget from $P_i^j$ to $N_j^i$ contribute
$g_{ij}$ to $P_i^j$; here $w_{ij}\in[n_0]$ and $g_{ij}\in[m_0]$,
because the contributions lie in $A_i^0$ and in $[m_0]$, respectively.
As the contributions of a choice gadget to its two endpoints sum to its parameter~$b_Q$,
each of these numbers also fixes the contribution to the other endpoint of its choice gadget.

The anchor root $N_i^j$ is an endpoint of the local gadget for $(i,j)$, of the
successor gadget for $(i,\operatorname{succ}_i^{-1}(j))$, and of the mirror
gadget
for $(j,i)$, and the anchor root $P_i^j$ is an endpoint of the three choice gadgets
for $(i,j)$.  As the selection is balanced, the three contributions at each
of these two anchor roots sum to the demand $d=n_0 z_0+m_0$.  Listing the
three summands in the order local, successor, mirror, these two \emph{demand
equations} read
\begin{align}
 \theta_{ij}z_0+\gamma_{ij}
   +(n_0-w_{i,\operatorname{succ}_i^{-1}(j)})z_0
   +(m_0-g_{ji})
     &= n_0 z_0+m_0
     && \text{(demand at $N_i^j$)},\label{eq:demand-at-n}\\
 (n_0 z_0+m_0-\theta_{ij}z_0-\gamma_{ij})
   +w_{ij}z_0+g_{ij}
     &= n_0 z_0+m_0
     && \text{(demand at $P_i^j$)}.\label{eq:demand-at-p}
\end{align}
Rearranged, \cref{eq:demand-at-n} reads
$(\theta_{ij}-w_{i,\operatorname{succ}_i^{-1}(j)})z_0=g_{ji}-\gamma_{ij}$ and
\cref{eq:demand-at-p} reads $(w_{ij}-\theta_{ij})z_0=\gamma_{ij}-g_{ij}$.  Each
right-hand side is a difference of two numbers in $[m_0]$ and thus has
absolute value at most $m_0-1<z_0$, whereas each left-hand side is a multiple
of $z_0$.  Both sides therefore vanish, which gives
\begin{align}\label{eq:consistency}
 \theta_{ij}=w_{ij}=w_{i,\operatorname{succ}_i^{-1}(j)}
 \quad\text{and}\quad
 \gamma_{ij}=g_{ij}=g_{ji}.
\end{align}
For fixed~$i$, the first equalities give
$\theta_{i,\operatorname{succ}_i(j)}=w_{ij}=\theta_{ij}$ for every $j\neq i$;
since $\operatorname{succ}_i$ is a cyclic permutation of $[p]\setminus\{i\}$,
all $\theta_{ij}$ with $j\neq i$ have the same value, which we denote by~$\theta_i$.
Consequently, all local gadgets for color $i$ choose the same vertex
$u^i_{\theta_i}$.  Applying the second equality in
\cref{eq:consistency} also with $(j,i)$ shows that the local gadgets for
$(i,j)$ and $(j,i)$ use the same edge identifier and hence the same edge,
which is therefore $\{u^i_{\theta_i},u^j_{\theta_j}\}$.  Hence, the chosen vertices
form a multicolored clique.

Conversely, suppose that $J$ has a multicolored clique with vertices $u^i_{\theta_i}$, $i\in[p]$, and let $\gamma_{ij}$ be the identifier of
$\{u^i_{\theta_i},u^j_{\theta_j}\}$.  For every ordered pair $(i,j)$, choose
the option of value $\theta_i z_0+\gamma_{ij}$ in the local gadget,
the option of value $\theta_i z_0$ in the successor gadget, and
the option of value $\gamma_{ij}$ in the mirror gadget for $(i,j)$.
These options exist because these values
belong to $A_i^j$, $A_i^0$, and $[m_0]$, respectively.
Since $\gamma_{ij}=\gamma_{ji}$, substitution in
\cref{eq:demand-at-n,eq:demand-at-p} gives $d$ at both $N_i^j$ and
$P_i^j$.  Thus, every anchor root receives its demand and the selection is balanced.
\end{proof}
}

In both reductions, a solution will cut every choice gadget by removing its chosen option,
and the option property translates contributions into numbers of vertices.

\begin{observation}\label{obs:selection-components}
  Suppose that every choice gadget of the choice graph has the $(\lambda,c)$-option property, and let $(o_Q)_{Q\in\mathcal Q}$ be a selection.
  After removing all options~$o_Q$, every anchor root~$v$ lies in its own connected component,
  which contains exactly $\lambda s_v+3c$ inner vertices of choice gadgets, where $s_v$ is the sum of the contributions of the three choice gadgets incident with~$v$.
  In particular, all these components contain exactly $\lambda d+3c$ inner vertices if and only if the selection is balanced.
\end{observation}
\begin{proof}
  Since choice gadgets pairwise intersect at most in a common endpoint,
  the component of~$v$ consists of~$v$ and, for each of the three choice gadgets~$Q$ incident with~$v$, the component of $Q-o_Q$ containing~$v$.
  By the option property, the latter does not contain the other endpoint of~$Q$,
  and it contains exactly $\lambda a+c$ inner vertices if $Q$ contributes~$a$ to~$v$.
\end{proof}

\subsection{Hardness for Cluster Editing}\label{sec:ce-hardness}
\toappendix{\sv{\subsection{Hardness for Cluster Editing}\label{app:ce-hardness}}}

\begin{restatable}[\appmark]{theorem}{ceHardness}\label{thm:ce-hardness}
\CEwFD{} is \(\mathrm{W}[1]\)-hard parameterized by the deletion budget
$\ell$.  The hardness also holds for the combined parameter
\(
  \ell+\pw(G)+\fes(G).
\)
\end{restatable}

For \CEwFD{}, the choice gadgets are paths with attached leaves, and their options are edges.

\begin{construction}[Choice gadget for \cref{thm:ce-hardness}]
  \label{constr:ce-choice-gadget}
  For $A=\{a_1, \ldots, a_T\}\subseteq[b]$ with $a_1 < \cdots < a_T$,
  an \emph{$(A,b)$-choice gadget} from a vertex $x$ to a vertex $y$ is constructed as follows.
  Create a path $(v_1, \ldots, v_{T+1})$ with $v_1=x$ and $v_{T+1}=y$.
  Attach
  $a_1$ new leaves to $v_1$,
  $a_t-a_{t-1}-1$ new leaves to $v_t$ for $t\in\{2,\ldots,T\}$,
  and $b-a_T$ new leaves to $v_{T+1}$.
  We call $v_1,\ldots,v_{T+1}$ the \emph{spine vertices} and the
  edges of $(v_1,\ldots,v_{T+1})$ the \emph{spine edges} of the choice gadget.
  The options are the spine edges, where $o_t\coloneqq\{v_t,v_{t+1}\}$ has value~$a_t$ for $t\in[T]$.
\end{construction}

The choice gadget has $b$ inner vertices;
deleting option~$o_t$ leaves $a_t$ of them in the component of~$x$ and the other $b-a_t$ in the component of~$y$.
So it has the $(1,0)$-option~property.

The choice graph alone does not quite do the trick; we need to ensure that merging two anchor roots is unaffordable
by attaching many leaves to each anchor root, making it part of a large star.
The unaffordability is captured by \cref{lem:stars-not-joined} below.

\begin{restatable}[\appmark]{lemma}{starsNotJoined}\label{lem:stars-not-joined}
  \looseness=-1
  Let $G$ be a graph and let $\ell$, $r$, and $d$ be integers such that
  $\ell \le rd$ and let $h = 9rd$.
  Suppose that $G$ contains $r$ pairwise vertex-disjoint induced stars, each with $h$ vertices.
  Let $\tilde G$ be a cluster graph obtained from $G$ by at most $k = \ell + r \binom{h + d}{2}$ edge edits,
  of which at most $\ell$ are deletions.
  Then no cluster of $\tilde G$ contains more than one center of the stars.
\end{restatable}

\toappendix{
\sv{\starsNotJoined*}
\begin{proof}
  Denote the stars by $S_1, \ldots, S_r$ with centers $c_1, \ldots, c_r$.
  As each star has $h-1 = 9rd - 1 > rd \ge \ell$ leaves, not all edges between center and leaves can be deleted.
  Let~$I$ be the set of leaves whose incident edge is deleted.
  Then, each $V(S_j)\setminus I$ appears within a cluster of $\tilde G$, and we need to add edges within the set to make it a clique.
  Let $s_j \coloneqq |V(S_j)\setminus I|$ for each $j \in [r]$.
  Then the number of edges added to the cluster containing $c_j$ is
  \[
    \textstyle\binom{s_j}{2} - (s_j-1).
  \]
  Since these edges are a disjoint set for every $j \neq j' \in [r]$,
  we know that any solution removing the edges indicent to $I$, but no further edges incident to any center $c_j$
  adds at least
  \( \sum_{j \in [r]} \binom{s_j}{2} - (s_j - 1) \)
  edges.
  Moreover, since
  \[
    \sum_{j \in [r]} s_j = rh - |I| \ge rh - \ell \ge r (h - d),
  \]
  there are $x_1, \dots, x_j \in \mathbb R$ such that $\sum_{j \in [r]} x_j = 0$ and
  \[
    s_j = \frac{1}{r}(rh-|I|) + x_j \ge (h-d) + x_j.
  \]
  Let $y \coloneqq h - |I|/r \ge h-d$.
  Then, the number of added edges is at least
  \begin{align*}
    \sum_{j \in [r]} \binom{s_j}{2} - (s_j - 1)
    &= \frac{1}{2} \sum_{j \in [r]} \big( s_j^2 - 3s_j + 2 \big)\\
    &= \frac{1}{2} \sum_{j \in [r]} \big( y^2 + 2x_jy + x_j^2 - 3(y+x_j) + 2 \big)\\
    &= \frac{1}{2} r(y^2 - 3y + 2) + \frac{2y-3}{2} \sum_{j \in [r]} x_j + \frac{1}{2}\sum_{j \in [r]} x_j^2\\
    &= \frac{1}{2} r(y-1)(y-2) + \frac{1}{2}\sum_{j \in [r]} x_j^2\\
    &\ge \frac{1}{2} r(h-d-1)(h-d-2) + \frac{1}{2}\sum_{j \in [r]} x_j^2\\
    &\ge \frac{1}{2} r\big( (h-d)^2 - 3(h-d) + 2 \big)\\
    &= \frac{1}{2} r\big( h^2 - 2hd -3h + d^2 + 3d + 2\big)\\
    &\ge \frac{1}{2} \big( h^2 - 2hd - 3h\big)\\
    &\ge \frac{1}{2} \big( h^2 - 5hd\big)\\
    &\ge \frac{1}{2} rh^2 - \frac{45}{2}r^2d^2.
  \end{align*}

  Suppose now that there are two different indices $j,j' \in [r]$ such that $c_j$ and $c_{j'}$ belong to the same cluster of $\tilde G$.
  Then, $S_j \setminus I$ and $S_{j'}\setminus I$ are contained in this cluster of $\tilde G$ and all edges between them need to be added.
  In the worst case, the edge between $c_j$ and $c_{j'}$ exists in $G$.
  Thus, we have at least $(h-\ell)(h-\ell)-1 \ge (9rd-rd)\cdot (9rd-rd)-1 = 64r^2d^2-1$ added edges. 
  This yields a contradiction since
  \begin{align*}
    k&=\ell + r\binom{h+d}{2}\\
     &\le \frac12 r (h+d)^2+rd\\
     &\le \frac12 r h^2 + rhd +rd^2+rd\\
     &= \frac12 r h^2 + 9r^2d^2 +rd^2+rd\\
     &\le \frac12 r h^2 + 11r^2d^2\\
     &< \Big(\frac12rh^2 - \frac{45}{2}r^2d^2\Big) + 64r^2d^2-1.\qedhere
  \end{align*}
\end{proof}
}

Turning the anchor roots into big stars finishes our construction.

\begin{construction}
  \label{constr:ce-hardness}
  Given an instance of \MCC{} with graph $J$ and a partition of its vertices into $U_1, \dots, U_p$, each of size $n_0$, with $m_0 \coloneqq |E(J)|$,
  we construct an instance $(G^\star, k^\star, \ell^\star)$ of \CEwFD{} as follows.

  Apply \cref{constr:choice-graph} using \cref{constr:ce-choice-gadget} for the choice gadgets and the above instance.
  Define $h \coloneqq 9rd$.
  To each anchor root, we attach $h-1$ new leaves, its \emph{private leaves}; together with the anchor root, they form a star with $h$ vertices, the \emph{anchor block} of the anchor root.
  Call the resulting graph $G^\star$ with $n \coloneqq |V(G^\star)|$ vertices and $m \coloneqq |E(G^\star)|$ edges.
  Finally, set
  $\ell^\star \coloneqq q$ and $k^\star \coloneqq r\textstyle\binom{h+d}{2} - m + 2 \ell^\star$.
\end{construction}

Let us give some intuition for the correctness.
We have $rh$ vertices in anchor blocks.
Each local, successor, and mirror gadget has, respectively, $n_0z_0 + m_0$, $n_0z_0$, and $m_0$ non-endpoint vertices;
this sums to $p(p-1) \cdot 2d = rd$.
Therefore, we have $n = r(h+d)$ vertices.

After using the deletion budget, we have $m-\ell^\star$ remaining edges.
By \cref{lem:stars-not-joined}, each anchor root is in its own cluster, so we have at least $r$ clusters.
The remaining addition budget thus suffices to obtain a cluster graph only if after the deletions, we have exactly $r$ size-$(h+d)$ components which in turn is only possible if the deletion budget was used to obtain a balanced selection, allowing us to invoke \cref{lem:balanced-selection}.

\toappendix{
\begin{restatable}[\appmark]{observation}{equalSizeClusters}\label{obs:equal-size-clusters}
  A cluster graph with at most $k$ clusters and $ky$ vertices has at least $k\binom{y}{2}$ edges, with equality if and only if it has exactly $k$ clusters of size $y$.
\end{restatable}

\begin{proof}
Assume that the clusters are $C_1, \ldots, C_k$ with sizes $y+x_1, \ldots, y+x_k$, respectively, where $x_1, \ldots, x_k$ are integers (possibly negative).
Since $\sum_{i \in [k]} (y+x_i) = |V(G)|= ky$, we have $\sum_{i \in [k]} x_i =0$.
Therefore, the number of edges is
\begin{align*}
 \sum_{i \in [k]} \binom{y+x_i}{2}
        &= \frac12\sum_{i \in [k]} (y+x_i)(y+x_i-1) \\
        &= \frac12\sum_{i \in [k]} (y^2+2yx_i+x_i^2-y-x_i) \\
        &= \frac12(ky^2 + 2y\sum_{i \in [k]} x_i + \sum_{i \in [k]} x_i^2 -ky - \sum_{i \in [k]} x_i)\\
        &= \frac12(ky^2 -ky + \sum_{i \in [k]} x_i^2).
\end{align*}
Since $\sum_{i \in [k]} x_i^2$ is non-negative, this proves the lower bound.
Moreover, we have equality if and only if each $x_i$ is 0.
\end{proof}
}

\toappendix{

\sv{\ceHardness*}
\begin{proof}
  We give a parameterized reduction from \MCC, which is $\mathrm{W}[1]$-hard parameterized by the number of colors~\cite{FellowsHRV09}.
  Instances with $p<3$ can be decided in polynomial time and mapped to the fixed yes-instance consisting of one isolated vertex or to the fixed no-instance consisting of a three-vertex path, in both cases with zero budgets.
  We output the same fixed no-instance if some pair of color classes has no edge between them.
  Hence we may assume $p\geq3$ and that every pair of color classes has an edge.
  In particular, all sets used in the choice gadgets below are nonempty.
  We may further assume that all color classes have the same size $n_0:=\max_{i\in[p]}|U_i|$: adding a new isolated vertex to $U_i$ leaves $E(J)$ unchanged and creates no multicolored clique, since a multicolored clique has $p\geq3$ vertices and hence contains no isolated vertex.
  Padding every color class up to $n_0$ vertices in this way raises $|U|$ to $pn_0$.
  Every color class then supplies the same number of vertices, which is what lets us give every anchor root the same demand below.

  We apply \cref{constr:ce-hardness} to obtain an instance $(G^\star, k^\star, \ell^\star)$ of \CEwFD{}.
  This takes polynomial time:
  we have $r,q\in O(p^2)$, $z_0=2m_0+10\in O(m_0)$, and $d=n_0z_0+m_0\in O(n_0m_0)$, so $h=9rd\in O(p^2n_0m_0)$.
  Thus, $n=r(h+d)\in O(p^4n_0m_0)$ and $m=n+q-r$ (see the parameter bounds below) are polynomial in $n_0+m_0+p$,
  and the budgets $k^\star$ and $\ell^\star$ have polynomial encoding length.
  Next, we prove the equivalence of the two instances, before finally proving the parameter bounds.

\proofparagraph{Forward direction}
Suppose that $J$ has a multicolored clique.  Use the balanced selection
from \cref{lem:balanced-selection}, and delete its chosen option, which is a spine edge, in
each choice gadget. These $q$ deletions split every choice gadget into two connected
endpoint components, each attached to one endpoint.  Every vertex other than
an anchor root is in one such component or is a private leaf of an anchor root,
and every
choice gadget between its two anchor roots has been cut.  The resulting graph
therefore has exactly $r$ connected parts, one for each anchor root.  The
component at an anchor root consists of its $h$-vertex anchor block and the
$d$ non-endpoint vertices of choice gadgets allocated to it, so every part has size
$h+d$. Add all missing edges inside every part.  The number of additions is
\[
 r\binom{h+d}{2}-(m-q),
\]
because the $m-q$ undeleted edges all lie inside the parts.  The resulting
edit set has $q=\ell^\star$ deletions and
$r\binom{h+d}{2}-m+2q=k^\star$ edits in total.
Thus, the constructed instance is a \yes-instance.

\proofparagraph{Backward direction: anchor separation}
For the converse, let $F$ be an edit set for the constructed instance, and
let $\mathcal C$ be the partition of $V(G^\star)$ into the clusters of
$G^\star\triangle F$.  Put $\delta:=|F\cap E(G^\star)|\leq q$.

Every anchor block is a star with $h$ vertices whose center is its anchor
root, and the $r$ anchor blocks are pairwise disjoint.  We apply
\cref{lem:stars-not-joined} to $G^\star$ with these $r$ stars.  Its two numerical hypotheses
hold:  First, $q\leq2r$ by our construction and $d\geq
z_0\geq2$, so the deletion budget satisfies $\ell^\star=q\leq2r\leq rd$.
Second, the $q$ choice gadgets have pairwise disjoint and nonempty spine-edge sets,
so $m\geq q$, and hence
\[
 k^\star=r\binom{h+d}{2}-m+2\ell^\star
   \leq r\binom{h+d}{2}+\ell^\star.
\]
As $|F|\leq k^\star$ and
$\delta\leq\ell^\star$, \cref{lem:stars-not-joined} shows that no cluster of
$G^\star\triangle F$ contains two anchor roots.

\proofparagraph{Backward direction: cuts and equal cluster sizes}
The endpoints of every choice gadget are distinct anchor roots, and hence belong to
different clusters.  Every choice gadget therefore has a deleted spine edge.  The
choice gadgets have disjoint spine-edge sets, so $\delta\geq q$.  Since
$\delta\leq\ell^\star=q$, equality holds: exactly one spine edge of every
choice gadget is deleted, and no other edge of $G^\star$ is deleted.

After its unique deleted spine edge is removed, each choice gadget consists of two
connected components, one attached to each endpoint.  No other edge
is deleted, so every vertex other than an anchor root belongs to one of these
components or is a private leaf attached to its anchor root.  Hence
$G^\star-(F\cap E(G^\star))$ has exactly $r$ connected components and every component
contains exactly one anchor root.  Every
such component lies in one cluster, because none of its edges is deleted.  A
cluster cannot contain two components, since it would then contain their two
anchor roots.  Thus these $r$ components are precisely the clusters in
$\mathcal C$.  In particular, every cluster induces a connected subgraph of
$G^\star$.

The cluster graph $G^\star\triangle F$ arises from $G^\star$ by deleting $\delta=q$ edges
and adding $|F|-\delta$ edges, so it has
\[
 m-\delta+(|F|-\delta)=m+|F|-2q\leq m+k^\star-2q=r\binom{h+d}{2}
\]
edges.  It has $r$ clusters and $r(h+d)$ vertices.  By \Cref{obs:equal-size-clusters}, every cluster has size
exactly $h+d$: if the sizes differed, the cluster graph would have more than
$r\binom{h+d}{2}$ edges.  The cuts consequently
give exactly the intended allocations: no edge incident with a private
leaf was deleted, so the cluster of an anchor root contains its anchor block
of $h$ vertices and therefore contains exactly $d$ non-endpoint vertices
from its three incident choice gadgets.  For each choice gadget, the
number on the anchor-root side is the allocation determined by its unique cut.
Hence, the cuts form
a balanced selection, and
\cref{lem:balanced-selection} yields a multicolored clique in $J$. This
proves the equivalence between the source and target instances.

\proofparagraph{Parameter bounds}
It remains to bound the target parameter. The deletion budget is
$\ell^\star=q=3p(p-1)$.  Removing the $r$ anchor roots from
$G^\star$ leaves a disjoint union of paths with attached leaves and isolated
vertices, so $\fvs(G^\star)\leq r$.  A width-one path decomposition of this
remainder has a bag $\{u,v\}$ for each spine edge $\{u,v\}$, with these bags
ordered along the path.  For every edge $\{u,w\}$ to a leaf~$w$, insert the bag
$\{u,w\}$ next to the consecutive spine-edge bags containing $u$; at a path
endpoint, insert it before or after its incident spine-edge bag.  If a spine
consists of a single vertex $u$, list all bags $\{u,w\}$ consecutively, or use
the singleton bag $\{u\}$ if no leaf is attached to~$u$.  Bags for all other
isolated vertices are also singleton bags, and the decompositions of distinct
components can be concatenated.  Adding all $r$ anchor roots to every bag gives
$\pw(G^\star)\leq r+1$.
To bound the feedback edge number, observe that $G^\star$ is connected as it is the choice graph plus some attached leaves.
Moreover, since each anchor block has $h-1$ edges and each choice gadget has one more edge than it has non-endpoint vertices,
we have
\[
  m = r(h-1) + rd + q = n + q - r.
\]
Therefore, $\fes(G^\star) = m-n+1 = q-r+1$.

Substituting the bounds gives
\[
 \ell^\star+\pw(G^\star)+\fvs(G^\star)+\fes(G^\star)
 \leq q+(r+1)+r+(q-r+1)
 =2q+r+2=8p(p-1)+2.
\]
This completes the parameterized reduction.
\end{proof}
}

\subsection{Hardness for \boldmath$H$-free Editing}\label{sec:hfree-hardness}
\toappendix{\sv{\subsection{Hardness for \texorpdfstring{$H$}{H}-free Editing}\label{app:hfree-hardness}}}

We now turn to \HFEwFD{}, $H$ being a clique minus one edge.

\begin{restatable}[\appmark]{theorem}{hfreeHardness}\label{thm:hfree-hardness}
For every fixed $\eta\ge4$ and $H=K_\eta-e$, \HFEwFD{} is
$\mathrm{W}[1]$-hard parameterized by the deletion budget~$\ell$.  The
hardness also holds for the combined parameter
\(
 \ell+\pw(G)+\fvs(G).
\)
\end{restatable}

We again use the choice graph from \cref{sec:choice-graph} and follow the ideas of \cref{thm:ce-hardness}.
Lifting them to $H=K_\eta-e$ requires the following ingredients (see \cref{fig:hfree-construction}).
\begin{enumerate}
\item \emph{A universal clique.}
  We add a clique~$Z$ on $\eta-3$ vertices that is complete to the graph~$G_A$ consisting of the choice graph and the anchor blocks.
  As long as none of these edges are deleted, every induced~$P_3$ in~$G_A$ forms an induced~$H$ together with~$Z$;
  hence, $G_A$ has to be turned into a cluster graph, just as for \CEwFD{}.
\item \emph{Protecting~$Z$.}
  A further clique~$W$ on $2\ell^\star+\eta$ vertices that is complete to~$Z$ makes it impossible to delete edges inside~$Z$.
\item \emph{Detaching as a new way of cutting.}
  A solution may also delete edges between~$Z$ and a vertex~$x$ of~$G_A$.
  Such a \emph{detached} vertex no longer needs to be in a cluster with its neighbors in~$G_A$,
  so a single deletion may cut a choice gadget at~$x$ (details after \cref{constr:hfree-hardness}).
  We design the choice gadgets such that detaching one of their designated option vertices is the only affordable way to cut them.
\item \emph{Shields.}
  Detaching an anchor root would cut all of its choice gadgets at once.
  To prevent this, every anchor root receives $\ell^\star+1$ disjoint cliques on $\eta-2$ vertices, called \emph{shields},
  such that detaching it requires a deletion in every shield, which would be too expensive.
\item \emph{Back to cluster editing.}
  With these ingredients, the attached vertices induce a cluster graph, and the detached vertices interact only in a controlled way with its clusters (\cref{lem:hfree-deletion-properties}).
  Isolating the detached vertices yields a cluster graph that is close to~$G_A$ and to which \cref{lem:stars-not-joined} applies as in \cref{thm:ce-hardness}.
\end{enumerate}

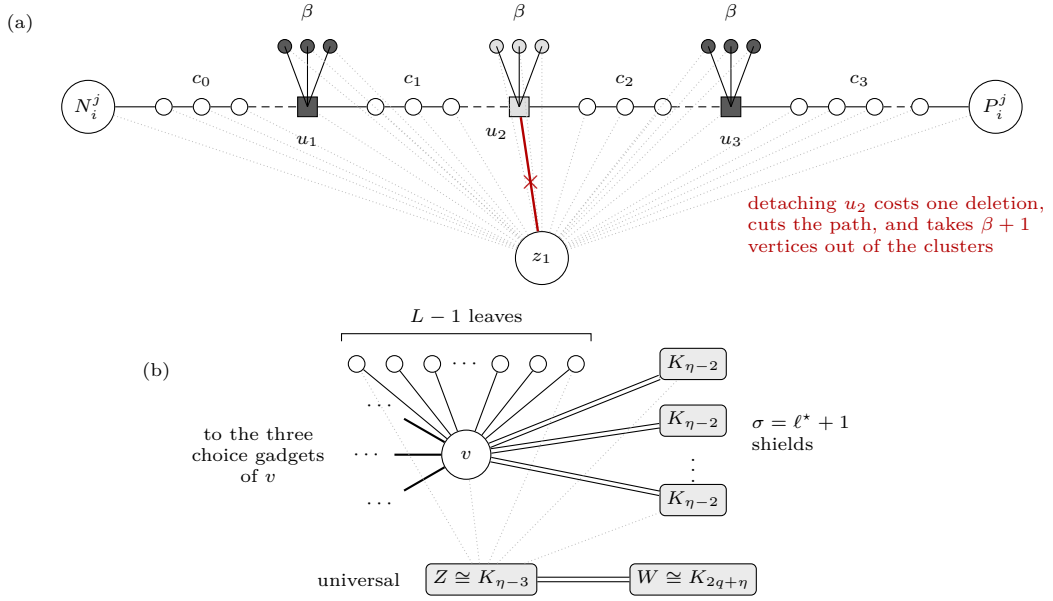
\begin{figure}[t]
\centering
\begin{tikzpicture}[
  anchorroot/.style={draw, circle, fill=white, minimum size=7mm, inner sep=1pt,
                     font=\scriptsize},
  option/.style={draw, rectangle, fill=black!65, minimum size=2.6mm, inner sep=0pt},
  optiond/.style={draw, rectangle, fill=black!12, minimum size=2.6mm, inner sep=0pt},
  fil/.style={draw, circle, fill=white, minimum size=2.2mm, inner sep=0pt},
  leaf/.style={draw, circle, fill=black!65, minimum size=1.8mm, inner sep=0pt},
  leafd/.style={draw, circle, fill=black!12, minimum size=1.8mm, inner sep=0pt},
  zed/.style={gray!55, densely dotted, line width=.35pt},
  lab/.style={font=\scriptsize}
]
 \node[lab] at (-0.9,1.1) {(a)};
 \node[anchorroot] (x) at (0,0) {$N_i^j$};
 \node[anchorroot] (y) at (12,0) {$P_i^j$};
 \foreach \i/\p in {1/1.0,2/1.5,3/2.0}{\node[fil] (f\i) at (\p,0) {};}
 \node[option] (u1) at (2.9,0) {};
 \foreach \i/\p in {4/3.8,5/4.3,6/4.8}{\node[fil] (f\i) at (\p,0) {};}
 \node[optiond] (u2) at (5.7,0) {};
 \foreach \i/\p in {7/6.6,8/7.1,9/7.6}{\node[fil] (f\i) at (\p,0) {};}
 \node[option] (u3) at (8.5,0) {};
 \foreach \i/\p in {10/9.4,11/9.9,12/10.4,13/11.0}{\node[fil] (f\i) at (\p,0) {};}
 \draw (x) -- (f1) -- (f2) -- (f3);
 \draw[densely dashed] (f3) -- (u1);
 \draw (u1) -- (f4) -- (f5) -- (f6);
 \draw[densely dashed] (f6) -- (u2);
 \draw (u2) -- (f7) -- (f8) -- (f9);
 \draw[densely dashed] (f9) -- (u3);
 \draw (u3) -- (f10) -- (f11) -- (f12);
 \draw[densely dashed] (f12) -- (f13);
 \draw (f13) -- (y);
 \foreach \i/\x in {1/2.9,3/8.5}{
   \foreach \j/\dx in {a/-.3,b/0,c/.3}{
     \node[leaf] (t\i\j) at (\x+\dx,.8) {};
     \draw (\x,0) -- (\x+\dx,.8);
   }
 }
 \foreach \j/\dx in {a/-.3,b/0,c/.3}{
   \node[leafd] (t2\j) at (5.7+\dx,.8) {};
   \draw (5.7,0) -- (5.7+\dx,.8);
 }
 \node[lab] at (2.9,1.25) {$\beta$};
 \node[lab] at (5.7,1.25) {$\beta$};
 \node[lab] at (8.5,1.25) {$\beta$};
 \node[lab] at (2.9,-.45) {$u_1$};
 \node[lab] at (5.4,-.35) {$u_2$};
 \node[lab] at (8.5,-.45) {$u_3$};
 \node[lab] at (1.5,.35) {$c_0$};
 \node[lab] at (4.3,.35) {$c_1$};
 \node[lab] at (7.1,.35) {$c_2$};
 \node[lab] at (10.2,.35) {$c_3$};
 \node[anchorroot] (z) at (6,-2.0) {$z_1$};
 \foreach \n in {x,f1,f2,f3,u1,f4,f5,f6,f7,f8,f9,u3,f10,f11,f12,f13,y}{
   \draw[zed] (z) -- (\n);}
 \foreach \n in {t1a,t1b,t1c,t2a,t2b,t2c,t3a,t3b,t3c}{
   \draw[zed] (z) -- (\n);}
 \draw[red!70!black, line width=.9pt] (z) -- (u2);
 \node[red!70!black] at (5.85,-1.0) {\large$\times$};
 \node[lab, red!70!black, align=left, anchor=west] at (8.6,-1.55)
   {detaching $u_2$ costs one deletion,\\ cuts the path, and takes
    $\beta+1$\\ vertices out of the clusters};
\begin{scope}[yshift=-4.6cm, xshift=5.0cm,
  root/.style={draw, circle, fill=white, minimum size=6.5mm, inner sep=1pt,
               font=\scriptsize},
  sleaf/.style={draw, circle, fill=white, minimum size=2.2mm, inner sep=0pt},
  blob/.style={draw, rounded corners=2pt, fill=black!8, inner sep=2.5pt,
               font=\scriptsize, align=center},
  cjoin/.style={double, line width=.4pt, double distance=1.3pt}]
  \node[lab] at (-4.1,1.1) {(b)};
  \node[root] (v) at (0,0) {$v$};
  \foreach \i/\x in {1/-1.45,2/-0.95,3/-0.45,4/0.45,5/0.95,6/1.45}{
    \node[sleaf] (l\i) at (\x,1.2) {};
    \draw (v) -- (l\i);
  }
  \node[lab] at (0,1.2) {$\cdots$};
  \draw (-1.65,1.5) -- (-1.65,1.6) -- (1.65,1.6) -- (1.65,1.5);
  \node[lab] at (0,1.82) {$L-1$ leaves};
  \foreach \a in {150,180,210}{
    \draw[line width=.8pt] (v) -- ++(\a:0.95);
    \node[lab] at (\a:1.3) {$\cdots$};
  }
  \node[lab, align=center] at (-2.75,0) {to the three\\ choice gadgets\\ of $v$};
  \node[blob] (s1) at (3.0,1.2) {$K_{\eta-2}$};
  \node[blob] (s2) at (3.0,0.45) {$K_{\eta-2}$};
  \node[lab] at (3.0,-0.1) {$\vdots$};
  \node[blob] (s3) at (3.0,-0.6) {$K_{\eta-2}$};
  \foreach \s in {s1,s2,s3}{\draw[cjoin] (v) -- (\s);}
  \node[lab, anchor=west, align=left] at (3.65,0.3) {$\sigma=\ell^\star+1$\\ shields};
  \node[blob] (Z) at (0.2,-1.65) {$Z\cong K_{\eta-3}$};
  \node[blob] (W) at (3.0,-1.65) {$W\cong K_{2q+\eta}$};
  \draw[cjoin] (Z) -- (W);
  \foreach \t in {v,l1,l6,s1,s3}{\draw[zed] (Z) -- (\t);}
  \node[lab, anchor=east] at (-0.75,-1.65) {universal};
\end{scope}
\end{tikzpicture}
\caption{(a)~An $(A,b)$-choice gadget for \HFEwFD{} with~$T=3$ option vertices (squares), its filler
segments (small circles, with the number of vertices indicated), and a
vertex $z_1$ of the clique~$Z$ (dotted edges).
Deleting the
edge $\{z_1,u_2\}$ detaches the grey option vertex~$u_2$: among the attached vertices, the leaves of~$u_2$ become isolated
and the two sides of the choice gadget fall apart;
the left side then contains $\lambda a_2+1$~inner vertices and the
right side $\lambda(b-a_2)+1$~inner vertices.
(b)~The anchor block of an anchor root $v$: a
star with $L-1$ leaves together with $\sigma=\ell^\star+1$ pairwise
disjoint shields (cliques on~$\eta-2$ vertices, each completely
joined to~$v$).
Dotted edges lead to the
clique~$Z$, which is complete to $V(G_A)$ and to the
clique~$W$; the latter keeps the edges inside $Z$ undeletable.}
\label{fig:hfree-construction}
\end{figure}

\begin{construction}[Choice gadget for \cref{thm:hfree-hardness}]
  \label{constr:hfree-choice-gadget}
  Let $\beta \coloneqq \max\{1, \eta-4\}$ and $\lambda \coloneqq \beta+2$.
  For $A=\{a_1, \ldots, a_T\}\subseteq[b]$ with $a_1 < \cdots < a_T$,
  an \emph{$(A,b)$-choice gadget} from a vertex $x$ to a vertex $y$ is constructed as follows (see \cref{fig:hfree-construction}(a)).
  Introduce $T$ \emph{option vertices} $u_1, \dots, u_T$.
  For each $t \in \{0, \dots, T\}$, create a \emph{filler segment} $P_t$, which is a path on $c_t$ \emph{filler vertices} with
  \begin{align}\label{eq:hfree-fillers}
    c_0:=\lambda a_1+1,\quad c_t:=\lambda(a_{t+1}-a_t-1)+1\ \ (t\in[T-1]),\quad c_T:=\lambda(b-a_T)+1.
  \end{align}
  Create a path, the \emph{spine}, by concatenating $x, P_0, u_1, P_1, \ldots, u_T, P_T, y$.
  Its vertices other than $x$ and~$y$ are the \emph{inner spine vertices}.
  Finally, attach $\beta$ new leaves to each option vertex~$u_t$.
  The options are the option vertices, where $o_t\coloneqq u_t$ has value~$a_t$ for $t\in[T]$.
\end{construction}
Recall that the inner vertices of the choice gadget are all its vertices other than $x$ and~$y$, including the leaves.
Removing option $o_t$ leaves $P_0,u_1,\ldots,u_{t-1},P_{t-1}$ and the leaves of $u_1,\ldots,u_{t-1}$ in the component of~$x$,
that is, $\lambda a_t+1$ inner vertices by \cref{eq:hfree-fillers}.
Symmetrically, the component of~$y$ has $\lambda(b-a_t)+1$ inner vertices.
So the choice gadget has the $(\lambda,1)$-option property.
Unlike for \CEwFD{}, $u_t$ and its $\beta$ leaves belong to neither component.
By \cref{obs:selection-components}, the anchor roots should thus receive
$\hat d\coloneqq\lambda d+3=\lambda(n_0z_0+m_0)+3$ inner vertices each.

\toappendix{
We start with some properties of the choice gadget for \HFEwFD{} (\cref{constr:hfree-choice-gadget}) and the resulting choice graph.
Since the elements of~$A$ are distinct integers in~$[b]$, we have~$c_t\ge1$ for all~$t\in\{0,\ldots,T\}$.
Consequently, no two option vertices are adjacent, no option vertex is
adjacent to $x$ or to~$y$, and $x$ and $y$ are not adjacent.

\begin{restatable}[\appmark]{observation}{hfreeChoiceGadget}\label{obs:hfree-choice-gadget}
Let $Q$ be an $(A,b)$-choice gadget from $x$ to $y$.  
Then the following hold.
\begin{romanenumerate}
\item\label{it:hfree-cg-size} $Q$ has exactly $\lambda(b+1)+1$ inner vertices.
\item\label{it:hfree-cg-option} For every $t\in[T]$, the graph $Q-u_t$ consists
      of (1)~the $\beta$ leaves of $u_t$ as isolated vertices, (2)~a connected
      component containing~$x$ with exactly $\lambda a_t+1$ inner
      vertices, and (3)~a connected component containing $y$ with exactly
      $\lambda(b-a_t)+1$ inner vertices.
\end{romanenumerate}
\end{restatable}

\begin{proof}
  Let us first prove (ii).
Deleting $u_t$ from the spine leaves two subpaths: one from $x$
through $P_0,u_1,\ldots,u_{t-1},P_{t-1}$, and one from $P_t$ through
$u_{t+1},\ldots,P_T$ to $y$.  
The leaves of~$u_t$ lose their only neighbor,
and every other leaf stays attached to its option vertex~$u_{t'}\ne u_t$.  
Hence, the connected component of~$x$ contains as inner vertices the segments
$P_0,\ldots,P_{t-1}$ and the option vertices $u_1,\ldots,u_{t-1}$ with their
leaves, that is,
\begin{align*}
 \sum_{t'=0}^{t-1}c_{t'}+(t-1)(1+\beta)
 &=\Bigl(\lambda a_1+1+\sum_{t'=1}^{t-1}\bigl(\lambda(a_{t'+1}-a_{t'}-1)+1\bigr)
   \Bigr)+(t-1)(\lambda-1)\\
 &=\lambda a_t+1
\end{align*}
vertices, where we used $\lambda=\beta+2$, \cref{eq:hfree-fillers}, and
that the sum telescopes to $\lambda a_t-\lambda(t-1)+t$.

For counting the number of vertices in the connected component containing $y$,
observe that an $(A, b)$-choice gadget from $x$ to $y$ is the same as an $(A', b)$-choice gadget
from $y$ to $x$, where $A' = \{b-a_T, b-a_{T-1}, \dots, b - a_1\}$, and the removed vertex is at the $t$-th \emph{last} option. 
Therefore, by symmetry, the component containing $y$ contains $\lambda(b-a_t)+1$ vertices.

Adding the $1+\beta$ vertices of $u_t$ and its leaves
gives $\lambda a_t+1+\lambda(b-a_t)+1+\lambda-1=\lambda(b+1)+1$ inner
vertices in total, which proves (i).
\end{proof}

By \cref{it:hfree-cg-size} of \cref{obs:hfree-choice-gadget}, the number of inner vertices of each local, successor, and mirror gadget is
$\lambda(n_0z_0 + m_0 + 1)+1$,
$\lambda(n_0z_0+1)+1$,
$\lambda(m_0+1)+1$, respectively,
which sums up to
\begin{align}\label{eq:hfree-inner-vertices}
  p(p-1) \cdot \big(\lambda(2n_0z_0+2m_0 + 3)+3\big) = r\lambda(n_0z_0+m_0) + q(\lambda+1).
\end{align}
Moreover, by \cref{obs:selection-components}, removing the options of a balanced selection leaves every anchor root in a component with exactly
$\hat d$ inner vertices, while the $q(\beta+1)$ removed option vertices and their leaves belong to none of these components.
}

\begin{construction}
  \label{constr:hfree-hardness}
  Given an instance of \MCC{} with graph $J$ and a partition of its vertices into $U_1, \dots, U_p$, each of size $n_0$, with
  $m_0 \coloneqq |E(J)|$,
  we construct an instance $(G^\star, k^\star, \ell^\star)$ of \HFEwFD{} as follows (see \cref{fig:hfree-construction}(b)).

  Apply \cref{constr:choice-graph} using \cref{constr:hfree-choice-gadget} for the choice gadgets and the above instance.
  Let $L \coloneqq 9r(\hat d + (q + 1)(\eta - 2))$.
  To each anchor root $v$, we add $L-1$ new leaves, each adjacent to $v$.
  Moreover, we add $\sigma \coloneqq q+1$ so-called \emph{shields}: pairwise disjoint cliques on $\eta-2$ new vertices that are all adjacent to $v$.
  The \emph{anchor block} of $v$ consists of $v$, its leaves and its shields.
  Hence, the anchor block has $h \coloneqq L + \sigma(\eta-2)$ vertices.
  Let $G_A$ be the graph formed by the choice graph plus the anchor blocks, and let $n_A$ and $m_A$ be the number of its vertices and edges.
  To $G_A$, add a set $Z$ of $\eta-3$ new vertices that form a clique and are all adjacent to all vertices in $G_A$.
  Further, add a set $W$ of $\omega \coloneqq 2q + \eta$ new vertices that form a clique and are all adjacent to all vertices in $Z$ (but have no neighbor in~$G_A$).
  Call the resulting graph $G^\star$.
  Finally,
  set 
  $\ell^\star \coloneqq q$
  and
  $k^\star \coloneqq r \textstyle\binom{h+\hat d}{2} - m_A + (\beta+3)q$.
\end{construction}

We now give details to Ingredient~3.
Consider a solution $F$ respecting the bounds $k^\star$ and $\ell^\star$ and let $\tilde G \coloneqq G^\star \triangle F$.
If $F$ contains no deletion incident with a vertex of $Z$, then $Z$ is a clique that is complete to $V(G_A)$ in $\tilde G$,
and $\tilde G$ is $H$-free if and only if $\tilde G[V(G_A)]$ is a cluster graph.
The solution may however delete edges between $Z$ and $G_A$ (and this is the intended mechanism in our reduction).
Hence, we record for each vertex $x \in V(G_A)$ the set of neighbors in $Z$ it loses in a solution,
that is, we let $Z_x \coloneqq \{ z \in Z \mid \{z, x\} \notin E(\tilde G)\}$.
We call a vertex $x \in V(G_A)$ \emph{detached} if $Z_x \ne \emptyset$, i.e., at least one edge between $x$ and $Z$ was removed,
and \emph{attached} otherwise.
Denote by $X$ the set of detached and by $Y \coloneqq V(G_A) \setminus X$ the set of attached vertices in $G_A$.
Solutions with at most $\ell^\star$ deletions have the following properties.

\begin{restatable}[\appmark]{lemma}{hfreeDeletionProperties}
  \label{lem:hfree-deletion-properties}
  Let $F \subseteq \binom{V(G^\star)}{2}$ have at most $\ell^\star$ deletions.
  If $\tilde G \coloneqq G^\star \triangle F$ is $H$-free, then, with $X$, $Y$, and the sets~$Z_x$ defined with respect to~$\tilde G$,
  \begin{romanenumerate}
    \item $Z$ induces a clique in $\tilde G$;
    \item $\tilde G[Y]$ is a cluster graph; and
    \item every detached vertex $x \in X$ is adjacent to at most $|Z_x|$ vertices of a cluster of $\tilde G[Y]$.
  \end{romanenumerate}
\end{restatable}

\toappendix{
  \sv{\hfreeDeletionProperties*}
  \begin{proof}
    Let $F^-$ be the set of deleted edges.
    (i) is clear for $\eta=4$.
    For $\eta \ge 5$, suppose there are $z, z' \in Z$ that are not adjacent in $\tilde G$.
    Recall that $W$ is a clique on $\omega = 2q + \eta = 2\ell^\star + \eta$ vertices.
    Since each edge in $F^-$ is incident with at most two vertices of $W$, at least $\omega - 2\ell^\star = \eta$ vertices of $W$ are incident with no edge in $F^-$.
    Let $W_0$ consist of $\eta-2$ such vertices.
    Then $W_0 \cup \{z, z'\}$ induces $H$, a contradiction.

    For (ii), suppose that $u, v, w \in Y$ induce a $P_3$ with center $v$.
    By definition of $Y$, $\tilde G$ contains all edges between $Y$ and $Z$;
    therefore, $Z \cup \{u, v, w\}$ induces $H$, a contradiction.

    For (iii), suppose that $x \in X$ is adjacent to at least $|Z_x|+1$ vertices of a cluster of $\tilde G[Y]$.
    Let $T$ be the set containing $x$ and any (exactly) $|Z_x|+1$ vertices of the cluster.
    Pick an arbitrary $z \in Z_x$ and let $Z' \coloneqq \{z\} \cup (Z \setminus Z_x)$.
    Then $|Z' \cup T| = (\eta - 2) - |Z_x| + |Z_x| + 2 = \eta$.
    In $\tilde G$, both $Z'$ and $T$ induce cliques, every edge between $Z'$ and $T\cap Y$ exists, and every edge between $x$ and $Z\setminus Z_x$ exists.
    The only non-edge in $\tilde G[Z' \cup T]$ thus is $\{z, x\}$; therefore $Z' \cup T$ induces $H$, a contradiction.
  \end{proof}
}

With this, the correctness follows the lines of \cref{thm:ce-hardness}\sv{; see \cref{app:hfree-hardness} for the full proof}.
First, no anchor root~$v$ is detached:
otherwise, as $v$ has more than $\ell^\star$ shields, some shield of~$v$ is not incident with any deletion,
and $v$ is adjacent to all of its $\eta-2>|Z_v|$ vertices, which are attached and lie in one cluster, contradicting \cref{lem:hfree-deletion-properties}(iii).
Second, isolating the detached vertices yields a cluster graph~$G^\dagger$;
as all detached vertices have small degree, $G^\dagger$ arises from~$G_A$ by few deletions,
and \cref{lem:stars-not-joined} shows that no cluster of~$G^\dagger$, and hence of~$\tilde G[Y]$, contains two anchor roots.
Thus, every choice gadget has to be cut, by detaching one of its inner spine vertices or by deleting a spine edge,
and as $\ell^\star=q$, every choice gadget is cut exactly once.
Third, cutting at a filler vertex or at a spine edge leaves more vertices in the components of the anchor roots than detaching an option vertex,
which also cuts off its $\beta$ leaves; by convexity, the addition budget then does not suffice.
Hence, the detached vertices are option vertices, one per choice gadget, and, as for \CEwFD{}, the budget forces all components of anchor roots to have the same size~$h+\hat d$.
By \cref{obs:selection-components}, the detached option vertices thus form a balanced selection,
and \cref{lem:balanced-selection} yields a multicolored clique.

\toappendix{
\sv{\hfreeHardness*}

\begin{proof}
  We give a parameterized reduction from \MCC.
  As in the proof of \cref{thm:ce-hardness}, we may assume that $p\ge3$, that every pair of color classes has an edge between them, and that all color classes have the same size~$n_0$;
  otherwise, we decide the instance in polynomial time or pad the color classes with isolated vertices,
  and in the former case we output a fixed \yes-instance (a single vertex) or a fixed \no-instance ($H$ itself), in both cases with zero budgets.
  We apply \cref{constr:hfree-hardness} to obtain an instance $(G^\star, k^\star, \ell^\star)$ of \HFEwFD{}.
  Since $\eta$ is a constant, this takes polynomial time:
  as before, $r,q\in O(p^2)$ and $z_0\in O(m_0)$;
  now $\hat d=\lambda(n_0z_0+m_0)+3\in O(n_0m_0)$ and $h=L+\sigma(\eta-2)=9r(\hat d+(q+1)(\eta-2))+(q+1)(\eta-2)\in O(p^2n_0m_0+p^4)$.
  Hence, by \cref{eq:hfree-inner-vertices}, $G_A$ has $n_A=rh+r\lambda d+q(\lambda+1)=r(h+\hat d)+q(\beta+1)$ vertices and $G^\star$ has $n_A+|Z|+|W|=n_A+O(p^2)$ vertices, which is polynomial in $n_0+m_0+p$,
  and the budgets have polynomial encoding length.
  Next, we prove the equivalence of the two instances, before finally proving the parameter bounds.

\proofparagraph{Forward direction}
Suppose $J$ has a multicolored clique with vertices~$u^i_{\theta_i}$, $i\in[p]$.
Then, by \cref{lem:balanced-selection}, the choice graph admits a balanced selection $(u_Q)_{Q \in \mathcal Q}$,
where $\mathcal Q$ is the set of choice gadgets and $u_Q$ is the selected option vertex in choice gadget $Q$.  
Fix an arbitrary vertex $z \in Z$.
Let $F^- \coloneqq \{ \{z, u_Q\} \mid Q \in \mathcal Q \}$,
let $X$ be the set of the $q$~selected option vertices (these are detached),
and let $Y \coloneqq V(G_A) \setminus X$ be the attached vertices.
For each anchor root $v$, let $C_v$ be the connected component of $G_A-X$ containing $v$.
Since the selection is balanced, we have $|C_v|=h + \hat d$ by \cref{obs:selection-components}.
Our edit set $F$ comprises $F^-$ and all non-adjacent pairs inside the sets $C_v$.
Let $\tilde G \coloneqq G^\star \triangle F$.

Clearly, $F^-$ consists of $\ell^\star=q$~edges.
For the number of additions, recall that~$G_A$ has~$m_A$ edges.
Every edge of~$G_A$ lies inside some set~$C_v$, except for the $q(\beta+2)$ edges incident with the option vertices in~$X$ (two spine edges and $\beta$ edges to leaves each);
note that the leaves of the option vertices in~$X$ form singleton components of $G_A-X$.
Hence, $F$ contains exactly $r\binom{h+\hat d}{2}-m_A+q(\beta+2)$~additions, and $|F|=q+r\binom{h+\hat d}{2}-m_A+q(\beta+2)=k^\star$.
It remains to verify that~$\tilde{G}$ is $H$-free.

Observe that~$\tilde{G}[Y]$ is a cluster graph, and that the neighborhood of every~$x\in X$ in~$\tilde{G}[V(G_A)]$ is an independent set.
Clearly, any forbidden induced subgraph~$R$ (equal to~$H$) cannot contain a vertex of~$W$ and of~$V(G_A)$ since these vertices are pairwise non-adjacent and their common neighborhood is~$Z$ and~$|Z|=\eta-3$.
Moreover, since~$Z\cup W$ induces a clique, we conclude that~$R$ contains only vertices of~$Z\cup V(G_A)$.
Clearly, $R$ needs to contain at least 3~vertices of~$V(G_A)$.
By~$S$ we denote the set of these vertices.

If~$S\subseteq Y$, the unique non-adjacent vertex pair of~$R$ has to have both endpoints in~$S$, and these two vertices together with any further vertex of~$S$ form a~$P_3$, a contradiction to the fact that~$\tilde{G}[Y]$ is a cluster graph.
Otherwise, $S$ contains a vertex~$x\in X$.
Suppose that $|S|\ge4$.
Since $R$ has only one non-adjacent pair, $x$ is adjacent to at least $|S|-2\ge2$ vertices of~$S$.
These vertices are pairwise non-adjacent, since the neighborhood of~$x$ in~$\tilde G[V(G_A)]$ is an independent set;
hence, they contain the only non-adjacent pair of~$R$, and $x$ is adjacent to all of $S\setminus\{x\}$.
But then $x$ has at least three pairwise non-adjacent neighbors in~$R$, a contradiction.
Thus, $|S|=3$, and since $|Z|=\eta-3$, $R$ contains all of~$Z$, in particular~$z$.
As $\{x,z\}$ is a non-adjacent pair, it is the only one; so $x$ is adjacent to both vertices of~$S\setminus\{x\}$,
which are therefore non-adjacent, a contradiction. 
Thus, $\tilde{G}$ is $H$-free.

\proofparagraph{Backward direction}
Let $F$ be a set of at most $k^\star$ edits with at most $\ell^\star=q$
deletions such that $\tilde G:=G^\star\triangle F$ is $H$-free.
We write~$F^-:=F\cap E(G^\star)$ and $F^+:=F\setminus E(G^\star)$.
By \cref{lem:hfree-deletion-properties}(i), $Z$ is a clique of $\tilde G$.
Hence, $X$, $Y$, and the sets $Z_x$ are defined.
Put $F_A:=F^-\cap\binom{V(G_A)}2$.
Every $x\in X$ contributes a deletion $\{z,x\}$ with $z\in Z_x$.
Note that these are distinct for distinct $x$ and distinct from the pairs of $F_A$, so
\begin{align}\label{eq:hfree-deletions}
 |X|+|F_A|\le|F^-|\le q .
\end{align}
By \cref{lem:hfree-deletion-properties}(ii), $\tilde G[Y]$ is a cluster graph.
Let~$\mathcal C$ be its set of clusters.

\smallskip
\emph{Step 1: no anchor root is detached.}
Suppose the anchor root $v$ lies in $X$. 
Then, $1\le|Z_v|\le\eta-3$.  
Observe that, since $v$ has more than $\ell^\star$ shields attached to it,
there is at least one shield that is not incident to any edge in $F^-$.
But then every vertex in the shield is attached; thus $v$ is adjacent to $\eta-2 > |Z_v|$ vertices of a cluster of $\tilde G[Y]$,
which is a contradiction to \cref{lem:hfree-deletion-properties}(iii).

\smallskip
\emph{Step 2: no cluster contains two anchor roots.}
Let $G^\dagger$ be the cluster graph on $V(G_A)$ consisting of
$\tilde G[Y]$ and the vertices of $X$ as isolated vertices.  It arises
from $G_A$ by deleting all edges incident with $X$ and all edges of
$F_A$ and adding the pairs of $F^+$ inside $Y$.
Every vertex of~$G_A$ other than an anchor root has degree at most~$\beta+2$ in~$G_A$:
leaves have degree~$1$, filler vertices degree~$2$, option vertices degree~$\beta+2$,
and shield vertices degree~$\eta-2\le\beta+2$, as $\beta\ge\eta-4$.
Hence, by Step~1 and
\cref{eq:hfree-deletions} the number of deletions is at most
\[
 \sum_{x\in X}\deg_{G_A}(x)+|F_A|\le q(\beta+2)+q=q(\beta+3)=:\ell^\dagger,
\]
and the number of edits is at most $\ell^\dagger+|F^+|\le
\ell^\dagger+k^\star\le\ell^\dagger+r\binom{h+\hat d}2$.
Here, the last inequality uses $m_A\ge(\beta+3)q$, which holds since already the edges to the $L-1$ leaves of any anchor root give
$m_A\ge L-1\ge 9r\hat d-1\ge 9r(\beta+3)-1\ge\frac32r(\beta+3)=(\beta+3)q$, using $\hat d\ge\lambda z_0\ge\beta+3$.

Set~$B:=\hat d+\sigma(\eta-2)$. 
Our definition of $L$ yields $L=9rB$, and $h+\hat d=L+\sigma(\eta-2)+\hat d=L+B$.
Now, we apply \Cref{lem:stars-not-joined} to $G_A$, with the $r$~stars formed by the anchor roots and their $L-1$ leaves,
where $B$, $L$, and $\ell^\dagger$ take the roles of the lemma's parameters $d$, $h$, and $\ell$, respectively.  
The hypotheses of \Cref{lem:stars-not-joined} hold: the stars
are vertex-disjoint and have $L$ vertices each, their centers are pairwise non-adjacent, because every choice gadget has
an inner vertex at each end, and $L+B=h+\hat d$. 
Hence, the edit bound is $\ell^\dagger+r\binom{L+B}2$, and $\ell^\dagger\le rB$, because
$\ell^\dagger=q(\beta+3)=\frac32r(\beta+3)\le3r(\beta+2)$ and
$B\ge \hat d\ge12(\beta+2)$.  
\Cref{lem:stars-not-joined} yields that no cluster of
$G^\dagger$ contains two anchor roots.  
Since the clusters of $\mathcal{C}$ are clusters of $G^\dagger$, the same holds for $\mathcal{C}$.

\smallskip
\emph{Step 3: exactly one break per choice gadget, and nothing else.}
To satisfy Step~2, we need to \emph{break} each choice gadget~$Q$, that is, it must contain a detached inner spine vertex or a deleted spine edge.
Since there are~$q=\ell^\star$ choice gadgets and our deletion budget is exactly~$\ell^\star$, we conclude that every choice gadget is broken exactly once, by exactly one deletion, and
no other edge is deleted.  
In particular every detached vertex~$x$ is an
inner spine vertex detached from a single vertex of $Z$ implying~$|Z_x|=1$.

Let $\kappa$, $\pi$ and $\delta_1$ be the number of choice gadgets broken at an option
vertex, at a filler vertex and at a spine edge, respectively. 
 Since every choice gadget is broken exactly once, a detached vertex has degree~$\beta+2$ if it is an option vertex and degree~$2$ if it is a filler vertex. 
 Moreover, we have
\[
 \kappa+\pi+\delta_1=q,\qquad |X|=\kappa+\pi,\qquad
 \sum_{x\in X}\deg_{G_A}(x)=\kappa(\beta+2)+2\pi .
\]
Removing the breaks splits $G_A[Y]$ into (a)~one component $C_v$ per anchor
root $v$, consisting of the anchor block of $v$ and the parts of its
three choice gadgets that stay connected to $v$, and (b)~the isolated leaves of the
$\kappa$ detached option vertices.  
An option vertex break removes $\beta+1$~vertices from the anchor components, a
filler break removes one, and an edge break removes none.  
Using
$n_A=r(h+\hat d)+q(\beta+1)$ (see \cref{eq:hfree-inner-vertices}), the
anchor components therefore contain
\[
 \sum_v|C_v|=n_A-\kappa(\beta+1)-\pi
 =r(h+\hat d)+(\beta+1)(q-\kappa)-\pi
 =r(h+\hat d)+(\beta+1)\delta_1+\beta\pi
\]
vertices, that is, $\tau:=(\beta+1)\delta_1+\beta\pi\ge0$ vertices more
than the $r(h+\hat d)$ of the intended solution.

Note that a positive surplus contradicts the edit budget:  On the one
hand, the anchor components lie in $r=2p(p-1)$~distinct clusters (Step~2) with
at least $r(h+\hat d)+\tau$ vertices in total, so by convexity completing them
requires at least $r\binom{h+\hat d+\tau/r}2\ge r\binom{h+\hat d}2+\tau(h+\hat d-\frac12)$~edges.  
On the other hand, $\tilde G[Y]$ has at most
$m_A-\sum_{x\in X}\deg_{G_A}(x)+k^\star-q
=r\binom{h+\hat d}2+(\beta+2)\delta_1+\beta\pi\le r\binom{h+\hat d}2+\frac32\tau$~edges, since the remaining budget for additions is $k^\star-q$.  
Hence, $\tau(h+\hat d-\frac12)\le\frac32\tau$, and as $h+\hat d>2$ this forces $\tau=0$.
Consequently, we have~$\delta_1=\pi=0$ and thus, every choice gadget is broken at
one of its option vertices.
In other words, the $q$~deletions are exactly the edges that join one option vertex~$u_Q$ of every choice gadget~$Q$ to some vertex of~$Z$.

\smallskip
\emph{Step~4: every cluster containing an anchor root has size $h + \hat d$.}  
By Step~3, the deletions are exactly the
$q$ edges joining an option vertex~$u_Q$ to a vertex of~$Z$ (one per choice gadget).
Hence, $\sum_v|C_v|=r(h+\hat d)$ and~$G_A[Y]$ has
$m_A-q(\beta+2)$~edges. 
Since at most
$k^\star-q=r\binom{h+\hat d}{2}-m_A+q(\beta+2)$ edges are added, we have~$|E(\tilde G[Y])|\le r\binom{h+\hat d}{2}$.
Moreover, each $C_v$ is connected in $\tilde G[Y]$ and thus
contained in a cluster $C^\ast_v$, and these clusters are distinct by
Step~2.  By the convexity of $x\mapsto\binom x2$ and
$\sum_v|C_v|=r(h+\hat d)$, we have
\[
 |E(\tilde G[Y])|\ \ge\ \sum_v\binom{|C^\ast_v|}2
 \ \ge\ \sum_v\binom{|C_v|}2\ \ge\ r\binom{h+\hat d}2 .
\]
So equality holds throughout: $C^\ast_v=C_v$, all $C_v$ have equal size yielding~$|C_v|=h + \hat d$, for every anchor root $v$.

\smallskip
\emph{Step 5: reading off the multicolored clique.}
By Step~3, $C_v$ contains the whole anchor block of $v$, which has $h$
vertices, so the remaining $|C_v|-h=\hat d$ vertices of $C_v$ are inner
vertices of choice gadgets which belong to the three choice gadgets at $v$.  
By Step~3 each of these choice gadgets
$Q$ is broken by detaching one of its option vertices $u_Q$. 
Hence, $(u_Q)_{Q\in\mathcal Q}$ is a selection, and $C_v$ is the component of~$v$ after removing its options.
As $C_v$ contains exactly $\hat d=\lambda d+3$ inner vertices of choice gadgets for every anchor root~$v$,
the selection is balanced by \cref{obs:selection-components}, and
\cref{lem:balanced-selection} yields a multicolored clique in~$J$.

\proofparagraph{Parameter bounds}
Recall that the deletion budget is~$\ell^\star=q=3p(p-1)$.

For the feedback vertex number, delete from $G^\star$ the $\eta-3$
vertices of $Z$, all but two vertices of $W$, the $r$ anchor roots, and
$\eta-4$ vertices of each of the $r\sigma$ shields.  What remains of $W$
is an edge and also what remains of any shield is an edge with no further incident edge, since its anchor root is gone.
Moreover, the leaves attached to the anchor roots become isolated. 
Finally, the choice graph becomes a disjoint union
of paths with attached leaves.  
Thus,
\[
 \fvs(G^\star)\le(\eta-3)+(\omega-2)+r+r\sigma(\eta-4)
 =r(\ell^\star+1)(\eta-4)+r+2\ell^\star+2\eta-5 .
\]

For the pathwidth, first build a path decomposition of $G_A$ in which
every bag contains all $r$ anchor roots.  For the anchor block of $v$
take one bag $\{u\}$ per leaf $u$ and, for every shield, one bag containing its vertices. 
For a choice gadget take the bag $\{u,w\}$ of every spine edge, ordered along the
spine, with the bags of the edges to the leaves of an option vertex inserted between
the two spine-edge bags at that vertex.  Concatenate all groups in any
order.  
This yields a path decomposition of~$G_A$ with bags of size at most
$r+\max\{\eta-2,2\}=r+\eta-2$.  
Add the $\eta-3$ vertices of $Z$ to every
bag and append one bag $Z\cup W$. 
The result is a path decomposition of
$G^\star$ with bags of size at most $\max\{r+2\eta-5,\ \omega+\eta-3\}
=2\ell^\star+2\eta-3$, using $2\ell^\star=3r$.  Hence
\[
 \pw(G^\star)\le2\ell^\star+2\eta-4=6p(p-1)+2\eta-4.
\]
Since $\eta$ is a constant, and~$r=2p(p-1)$, we observe that~$\ell^\star+\pw(G^\star)+\fvs(G^\star)\in O(p^4)$.
\end{proof}
}

\section{FPT-algorithms wrt. Deletions \boldmath$\ell$ Combined with Above-guarantees}\label{sec:above-guarantees}
\toappendix{\sv{\section{Missing Details for the FPT-Algorithms}\label{app:fpt}}}

\subsection{Cluster Editing with Few Deletions}
We now prove the following.

\begin{restatable}[\appmark]{theorem}{fptMuEll}\label{thm:fpt-mu-ell}
  An instance of \CEwFD{} given together with a modification-disjoint
 collection $\mathcal P$ of induced $P_3$'s of the input graph can be decided in $O(4^{\mu+\ell}\cdot(n+m+\mu))$ time, where $\mu=k-|\mathcal P|$.
\end{restatable}

The basic strategy is to first tentatively add all missing edges of the $P_3$s in the packing.
If this gives a cluster graph, we have found a solution.
Otherwise, we can find an induced $P_3$, say $Q$, in the augmented graph.
As we will see, there is always a solution that either adds the missing edge, yielding excess over the packing, or performs an edge-deletion in the neighborhood of $Q$.

For an induced $P_3$~$Q$ of~$G$ with edges $\{p,q\}$ and $\{q,r\}$, the definitions in \cref{sec:preliminaries} give $\miss(Q)=\{p,r\}$ and $\mods(Q)=\{\{p,q\},\{q,r\},\{p,r\}\}$; these are precisely the three vertex pairs whose toggling destroys~$Q$.
In the following we assume that, along with the instance, we are given a modification-disjoint collection $\mathcal P$ of induced $P_3$'s of the input graph, and we define the above-guarantee parameter $\mu\coloneqq k-|\mathcal P|$.
Clearly, if $\mu<0$, then we face a no-instance and we thus assume $\mu\geq0$.
We may also assume that $\mathcal P$ is inclusion-maximal: further induced $P_3$'s whose modification sets are disjoint from those of the current members can be added greedily, each decreasing~$\mu$ by one.

For a modification-disjoint collection $\mathcal P$ of induced $P_3$'s of $G$,
let $\miss(\mathcal P)\coloneqq\{\miss(Q)\mid Q\in\mathcal P\}$ collect the
missing pairs of the members of $\mathcal P$ and let
\begin{align}\label{eq:completed-graph}
 G^{\mathcal P}\coloneqq(V(G),E(G)\cup \miss(\mathcal P))
\end{align}
be the graph obtained from $G$ by adding all of them.
Distinct members of $\mathcal P$ have disjoint modification sets and hence distinct missing pairs, so $|{\miss(\mathcal P)}|=|\mathcal P|$.
Adding $\miss(\mathcal P)$ is thus a cheapest way of destroying all members of $\mathcal P$, and it is the only such way that spends no deletion.
We treat it as the default and deviate from it only when forced to.

We call a tuple $(G,k,\ell,\mathcal P)$ a \emph{packed instance} if $(G,k,\ell)$ is an instance of \CEwFD{} and $\mathcal P$ is a modification-disjoint collection of induced $P_3$'s of $G$, and we call it a \emph{yes-instance} if $(G,k,\ell)$ is one.

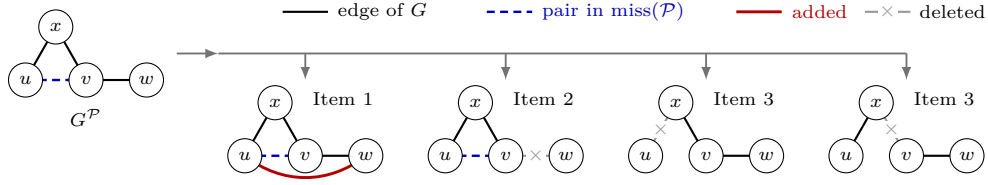
\begin{figure}[t]
\centering
\begin{tikzpicture}[
  vtx/.style={draw, circle, fill=white, minimum size=4.5mm, inner sep=0pt,
              font=\scriptsize},
  gedge/.style={black, line width=.8pt},
  added/.style={blue!70!black, densely dashed, line width=.9pt},
  newpair/.style={red!70!black, line width=1.1pt},
  gone/.style={black!40, line width=.8pt},
  xmark/.style={midway, inner sep=.5pt, fill=white, font=\scriptsize},
  lbl/.style={font=\scriptsize},
  flow/.style={black!55, line width=.7pt, -latex}
]
  \newcommand{\panelnodes}{%
    \node[vtx] (u) at (0,0) {$u$};
    \node[vtx] (x) at (.4,.7) {$x$};
    \node[vtx] (v) at (.8,0) {$v$};
    \node[vtx] (w) at (1.6,0) {$w$};}

  \begin{scope}[shift={(-2.9,1)}]
    \draw[gedge] (0,0) -- (.4,.7) -- (.8,0) (.8,0) -- (1.6,0);
    \draw[added] (0,0) -- (.8,0);
    \panelnodes
    \node[lbl] at (.8,-.5) {$G^{\mathcal P}$};
  \end{scope}

  \draw[flow] (-.9,1.35) -- (-.35,1.35);
  \draw[gedge, black!55, line width=.7pt] (-.35,1.35) -- (8.75,1.35);
  \foreach \cx in {.8,3.45,6.1,8.75} \draw[flow] (\cx,1.35) -- (\cx,1.0);

  \begin{scope}[shift={(0,0)}]
    \draw[gedge] (0,0) -- (.4,.7) -- (.8,0) (.8,0) -- (1.6,0);
    \draw[added] (0,0) -- (.8,0);
    \draw[newpair] (0,0) to[bend right=38] (1.6,0);
    \panelnodes
    \node[lbl] at (1.3,.75) {\cref{item:rule-add}};
  \end{scope}

  \begin{scope}[shift={(2.65,0)}]
    \draw[gedge] (0,0) -- (.4,.7) -- (.8,0);
    \draw[gone] (.8,0) -- (1.6,0) node[xmark] {$\times$};
    \draw[added] (0,0) -- (.8,0);
    \panelnodes
    \node[lbl] at (1.3,.75) {\cref{item:rule-delete}};
  \end{scope}

  \begin{scope}[shift={(5.3,0)}]
    \draw[gone] (0,0) -- (.4,.7) node[xmark] {$\times$};
    \draw[gedge] (.4,.7) -- (.8,0) (.8,0) -- (1.6,0);
    \panelnodes
    \node[lbl] at (1.3,.75) {\cref{item:rule-added-pair}};
  \end{scope}
  \begin{scope}[shift={(7.95,0)}]
    \draw[gedge] (0,0) -- (.4,.7) (.8,0) -- (1.6,0);
    \draw[gone] (.4,.7) -- (.8,0) node[xmark] {$\times$};
    \panelnodes
    \node[lbl] at (1.3,.75) {\cref{item:rule-added-pair}};
  \end{scope}

  \begin{scope}[shift={(0,3.4)}]
    \draw[gedge] (.5,-1.5) -- (1.1,-1.5) node[lbl, right] {edge of $G$};
    \draw[added] (3.2,-1.5) -- (3.8,-1.5)
    node[lbl, right] {pair in $\miss(\mathcal P)$};
    \draw[newpair] (6.5,-1.5) -- (7.1,-1.5) node[lbl, right] {added};
    \draw[gone] (8.2,-1.5) -- (8.8,-1.5) node[xmark] {$\times$};
    \node[lbl, right] at (8.8,-1.5) {deleted};
  \end{scope}
\end{tikzpicture}
\caption{\Cref{brule:destroy-p3} applied to the induced $P_3$ on $u$, $v$, $w$
of $G^{\mathcal P}$, in the case that $\{u,v\}$ is the missing pair of the
member $Q$ of $\mathcal P$ on $u$, $x$, $v$, and that $\{v,w\}$ is an edge of
$G$.  Keeping $\{u,v\}$ out removes $Q$ from the collection, so the pair is no
longer added and $Q$ has to be destroyed by one of its own two edges; this is
the item that produces two packed instances rather than one.}
\label{fig:branching-rule}
\end{figure}

\looseness=-1
The rule below destroys a fixed induced $P_3$ of $G^{\mathcal P}$ with edges
$\{u,v\}$ and $\{v,w\}$.  A cluster graph contains no induced $P_3$, so a
solution adds $\{u,w\}$ or keeps one of $\{u,v\}$ and $\{v,w\}$ out of the
cluster graph, and the rule branches on these three possibilities.  Adding
$\{u,w\}$ is an edit of $G$ outside $\mathcal{P}$ and is charged as one towards $\mu$.  For the other two the pair
kept out may or may not be an edge of $G$, and the case in which it is not does not immediately reduce the parameter.
To deal with that, say we keep out the nonedge $\{u,v\}$ of $G$.
There is an induced $P_3$ in $G$ that was responsible for $\{u,v\}$ being added to $G^{\mathcal{P}}$.
Keeping $\{u,v\}$ out now means that we also have to delete one of the other edges of this $P_3$, giving an additional deletion which decreases parameter $\ell$.
\Cref{fig:branching-rule} shows the resulting packed instances in this case.

Each produced packed instance thus records one edit, a toggled pair $f$, and it is given the collection $\mathcal P-f\coloneqq\{Q\in\mathcal P\mid f\notin\mods(Q)\}$ of those members that this edit leaves intact.

\begin{branchingrule}\label{brule:destroy-p3}
 Let $(G,k,\ell,\mathcal P)$ be a packed instance such that $G^{\mathcal P}$ is
 not a cluster graph, and let $u$, $v$, $w$ induce a $P_3$ in $G^{\mathcal P}$
 with edges $\{u,v\}$ and $\{v,w\}$.  Produce the following packed instances.
 \begin{enumerate}
  \item\label{item:rule-add} $(G+\{u,w\},\,k-1,\,\ell,\,\mathcal P)$.
  \item\label{item:rule-delete} For each $e\in\{\{u,v\},\{v,w\}\}$ with $e\in
   E(G)$, the packed instance $(G-e,\,k-1,\,\ell-1,\,\mathcal P-e)$.
  \item\label{item:rule-added-pair} For each $e\in\{\{u,v\},\{v,w\}\}$ with
   $e\notin E(G)$, let $Q_e$ denote the member of $\mathcal P$ with
   $\miss(Q_e)=e$ and produce, for each of the two edges
   $f\in\mods(Q_e)\cap E(G)$, the packed instance
   $(G-f,\,k-1,\,\ell-1,\,\mathcal P-f)$.
 \end{enumerate}
\end{branchingrule}

We next verify that the rule is well defined, decreases $\mu+\ell$, and
that it is safe.

\begin{restatable}[\appmark]{lemma}{bruleWellFormed}\label{lem:brule-well-formed}
  \cref{brule:destroy-p3} is well defined and produces at most five packed instances.
  Every produced packed instance $(G',k',\ell',\mathcal P')$ satisfies $(k'-|\mathcal P'|)+\ell'\leq(k-|\mathcal P|)+\ell-1$.
\end{restatable}
\toappendix{
  \sv{\bruleWellFormed*}
\begin{proof}
 We first check that the rule is well defined, that is, that $Q_e$ in
 \cref{item:rule-added-pair} exists and is unique.  If $e\in\{\{u,v\},\{v,w\}\}$
 and $e\notin E(G)$, then $e\in E(G^{\mathcal P})\setminus E(G)=\miss(\mathcal P)$
 by \cref{eq:completed-graph}, so some member of $\mathcal P$ has missing pair
 $e$; this member is unique because distinct members have disjoint modification
 sets and hence distinct missing pairs.

 Next we check that each produced tuple is a packed instance, that is, that its
 fourth component is a modification-disjoint collection of induced $P_3$'s of
 its first component.  Modification disjointness is inherited by subcollections,
 so only the induced-$P_3$ property needs an argument.  For a tuple with first
 component $G-f$ and fourth component $\mathcal P-f$, let $Q\in\mathcal P-f$.
 By the definition of $\mathcal P-f$ we have $f\notin\mods(Q)$, so $G-f$ and $G$
 agree on the three pairs in $\mods(Q)$ and $Q$ is an induced $P_3$ of $G-f$.
 For the tuple of \cref{item:rule-add}, let $Q\in\mathcal P$.  The two edges of
 $Q$ lie in $E(G)$ and its missing pair lies in $\miss(\mathcal P)$, so
 $\mods(Q)\subseteq E(G^{\mathcal P})$.  Since $u$, $v$, $w$ induce a $P_3$ in
 $G^{\mathcal P}$, we have $\{u,w\}\notin E(G^{\mathcal P})$ and hence
 $\{u,w\}\notin\mods(Q)$.  Thus $G+\{u,w\}$ and $G$ agree on $\mods(Q)$ and $Q$
 is an induced $P_3$ of $G+\{u,w\}$.

 The rule produces one packed instance in \cref{item:rule-add} and, for each of
 the two pairs $\{u,v\}$ and $\{v,w\}$, either one in \cref{item:rule-delete} or
 two in \cref{item:rule-added-pair}, hence at most five in total.

 It remains to bound $(k'-|\mathcal P'|)+\ell'$.  In \cref{item:rule-add} we
 have $k'=k-1$, $\mathcal P'=\mathcal P$, and $\ell'=\ell$, so this quantity is
 $(k-|\mathcal P|)+\ell-1$.  In \cref{item:rule-delete,item:rule-added-pair} an
 edge $f\in E(G)$ is deleted and thus $k'=k-1$ and $\ell'=\ell-1$.  If
 $f\in\mods(Q)$ for some $Q\in\mathcal P$, then this $Q$ is unique by
 modification disjointness, so $|\mathcal P'|=|\mathcal P|-1$ by
 the definition of $\mathcal P-f$ and the quantity is $(k-|\mathcal P|)+\ell-1$.
 Otherwise $|\mathcal P'|=|\mathcal P|$ and it is $(k-|\mathcal P|)+\ell-2$.  In
 either case it is at most $(k-|\mathcal P|)+\ell-1$.
\end{proof}
}

\begin{lemma}\label{lem:brule-safe}
 \cref{brule:destroy-p3} is safe.
\end{lemma}

\begin{proof}
 Each packed instance produced by the rule has the form $(G',k-1,\ell',\mathcal P')$ with $G'\coloneqq G\triangle\{f\}$ for a single pair~$f$, where $\ell'=\ell-1$ if $f\in E(G)$ and $\ell'=\ell$ otherwise: $f=\{u,w\}\notin E(G)$ in \cref{item:rule-add}, and $f\in E(G)$ in \cref{item:rule-delete,item:rule-added-pair}.
 The graphs $G$ and $G'$ differ only in the pair~$f$.
 Hence every $F'\subseteq\binom{V(G)}{2}$ satisfies $G'\triangle F'=G\triangle(F'\triangle\{f\})$, and every pair $g\neq f$ lies in $E(G)$ if and only if it lies in $E(G')$.

 \looseness=-1
 Suppose first that one of the produced packed instances is a yes-instance, and let $F'$ be a solution of $(G',k-1,\ell')$.
 Then $F\coloneqq F'\triangle\{f\}$ satisfies that $G\triangle F=G'\triangle F'$ is a cluster graph and $|F|\leq|F'|+1\leq k$.
 Every pair in $F\cap E(G)$ other than~$f$ lies in $F'\cap E(G')$, so $|F\cap E(G)|\leq\ell'+1=\ell$ if $f\in E(G)$ and $|F\cap E(G)|\leq\ell'=\ell$ otherwise.
 Thus $(G,k,\ell)$ is a yes-instance.

 \looseness=-1
 Suppose conversely that $(G,k,\ell,\mathcal P)$ is a yes-instance and let $F$ be a solution of $(G,k,\ell)$.
 We claim that $F$ contains the pair $f$ of one of the produced packed instances.
 This finishes the proof: $F'\coloneqq F\setminus\{f\}$ satisfies that $G'\triangle F'=G\triangle F$ is a cluster graph, $|F'|=|F|-1\leq k-1$, and $F'\cap E(G')=(F\cap E(G))\setminus\{f\}$, which has size $|F\cap E(G)|-1\leq\ell-1=\ell'$ if $f\in E(G)$ and size $|F\cap E(G)|\leq\ell=\ell'$ otherwise; so the packed instance for~$f$ is a yes-instance.
 It remains to prove the claim.
 The graph $G\triangle F$ is a cluster graph and thus contains no induced $P_3$, so $\{u,w\}\in E(G\triangle F)$ or $\{u,v\}\notin E(G\triangle F)$ or $\{v,w\}\notin E(G\triangle F)$.
 In the first case $\{u,w\}\notin E(G^{\mathcal P})\supseteq E(G)$ gives $\{u,w\}\in F$, the pair of \cref{item:rule-add}.
 In the remaining cases let $e\in\{\{u,v\},\{v,w\}\}$ with $e\notin E(G\triangle F)$.
 If $e\in E(G)$, then $e\in F$, the pair of \cref{item:rule-delete} for $e$.
 If $e\notin E(G)$, then $e\in \miss(\mathcal P)$ and $e\notin F$; since $Q_e$ is an induced $P_3$ in $G$ we have $F\cap\mods(Q_e)\neq\emptyset$, so $F$ contains one of the two edges of $Q_e$, and these are the pairs of \cref{item:rule-added-pair} for $e$.
\end{proof}

From the safety of \cref{brule:destroy-p3} and the fact that it is a 5-way branching that decreases $\mu + \ell$ by at least one in each case we would now straightforwardly get a $O(5^{\mu+\ell}\cdot(n+m+\mu))$-time algorithm.
With the following observation we can reduce the exponential part to $4^{\mu + \ell}$:
If $u, v, w$ induce a $P_3$ in which at most one of $\{u, v\}$ and $\{v, w\}$ is not in $E(G)$, then the branching factor is 4 instead.
Such a $P_3$ always exists:

\begin{lemma}\label{lem:p3-with-original-edge}
 Let $\mathcal P$ be a modification-disjoint collection of induced $P_3$'s of a
 graph $G$ and suppose that $G^{\mathcal P}$ is not a cluster graph.  Then
 $G^{\mathcal P}$ contains three vertices that induce a $P_3$ in which at least
 one of the two edges is an edge of $G$.
\end{lemma}

\begin{proof}
 Let $u$, $v$, $w$ induce a $P_3$ in $G^{\mathcal P}$ with edges $\{u,v\}$ and
 $\{v,w\}$.  If one of these two pairs lies in $E(G)$, then we are done, so
 assume that both lie in $\miss(\mathcal P)$.  Let $Q_1\in\mathcal P$ be the member
 with $\miss(Q_1)=\{u,v\}$ and let $x$ be its third vertex, and let
 $Q_2\in\mathcal P$ be the member with $\miss(Q_2)=\{v,w\}$ and let $y$ be its
 third vertex.  Then $\{u,x\},\{x,v\},\{v,y\},\{y,w\}\in E(G)$.

 The vertices $u$, $v$, $w$, $x$, $y$ are pairwise distinct.  Indeed,
 $x\notin\{u,v\}$ and $y\notin\{v,w\}$ because $Q_1$ and $Q_2$ are $P_3$'s;
 $x\neq w$ because $\{x,v\}\in E(G)$ whereas $\{w,v\}\notin E(G)$; $y\neq u$ by
 the symmetric argument; and $x\neq y$ because otherwise $\{x,v\}\in\mods(Q_1)\cap\mods(Q_2)$, contradicting modification disjointness.

 For the required $P_3$, recall that $\{u,w\}\notin E(G^{\mathcal P})$.
 If $\{u,y\}\notin E(G^{\mathcal P})$, then $u$, $v$, $y$ induce a $P_3$ with $\{v,y\}\in E(G)$.
 If $\{u,y\}\in E(G^{\mathcal P})$, then $u$, $y$, $w$ induce a $P_3$ with $\{y,w\}\in E(G)$.
\end{proof}

Combining \cref{lem:p3-with-original-edge,lem:brule-safe,lem:brule-well-formed} now gives \cref{thm:fpt-mu-ell}.

\toappendix{
  \sv{\fptMuEll*}
\begin{proof}
  Consider the recursive algorithm that, given a packed instance $(G,k,\ell,\mathcal P)$, proceeds as follows.
  If $k<|\mathcal P|$ or $\ell<0$, then it returns ``no''.
  Otherwise, if $G^{\mathcal P}$ is a cluster graph, then it returns ``yes''.
  Otherwise it finds a $P_3$ with at least one edge in $G$ as specified in \cref{lem:p3-with-original-edge}.
  It applies \cref{brule:destroy-p3} and returns ``yes'' if and only if it returns ``yes'' on at least one of the produced packed instances.

  \proofparagraph{Correctness} We proceed by induction on $(k-|\mathcal P|)+\ell$.
  If the algorithm returns ``no'' in the first step, then either $k<|\mathcal P|$, where ``no'' clearly is correct, or $\ell<0$, in which case no $F$ satisfies $|F\cap E(G)|\leq\ell$.
  If the algorithm returns ``yes'' in the second step, then $F\coloneqq \miss(\mathcal P)$ satisfies $G\triangle F=G^{\mathcal P}$, which is a cluster graph, as well as $|F|=|\mathcal P|\leq k$ because the first step did not apply, and $F\cap E(G)=\emptyset\leq\ell$; so $(G,k,\ell)$ is a yes-instance.
  Otherwise the algorithm applies \cref{brule:destroy-p3}.
  By \cref{lem:brule-well-formed} the produced tuples are packed instances with a smaller value of $(k-|\mathcal P|)+\ell$, so the induction hypothesis applies to them and the recursive calls answer them correctly; by \cref{lem:brule-safe} the answer for $(G,k,\ell,\mathcal P)$ agrees with the disjunction of these answers, which is what the algorithm returns.

 \proofparagraph{Running time}
 By \cref{lem:brule-well-formed} each application of \cref{brule:destroy-p3}
 produces decreases $(k-|\mathcal P|)+\ell$ by 
 at least one in each case and by choosing the $P_3$ according to \cref{lem:p3-with-original-edge} there are at most four recursive calls.
 The algorithm returns without recursing as soon as $k-|\mathcal P|<0$ or $\ell<0$.
 Hence the recursion tree has depth at most
 $\mu+\ell+1$ and $O(4^{\mu+\ell})$ nodes.
 The graph in a node arises from the
 input graph by at most $\mu+\ell$ edits and thus has at most $m+\mu+\ell$
 edges.
 Within a node we build $G^{\mathcal P}$ in $O(n+m+\mu)$ time, and we
 test whether $G^{\mathcal P}$ is a cluster graph and, if not, find an induced
 $P_3$ within the same time bound: compute the connected components, find a
 component $C$ containing a vertex $u$ of degree less than $|C|-1$, which exists
 if and only if some component is not a clique, and take a vertex $w\in C$ at
 distance two from $u$ together with a common neighbor $v$ of $u$ and $w$.
 Once we have found a $P_3$ we perform the case-distinction in \cref{lem:p3-with-original-edge} which can be done in $O(n + m + \mu)$ time and thus we obtain a $P_3$ satisfying \cref{lem:p3-with-original-edge} in the same time.
 This gives the claimed running time.
\end{proof}
}

\subsection{Generalization to\boldmath~$H$-free Editing}

Let $h \ge 3$ and $H=K_h - e$, that is, the graph on $h$ vertices with $\binom{h}{2}-1$ edges.
Given a modification-disjoint collection~$\mathcal H$ of induced copies of~$H$ in the input graph, with missing pairs and modification sets as defined in \cref{sec:preliminaries}, let $\mu\coloneqq k - |\mathcal{H}|$.
(Note that we overload $\mu$ here; in the context of \CE\ it instead means the excess over a packing of $P_3$s and it will always be clear what is meant from the context.)

\begin{restatable}[\appmark]{theorem}{hfreeFptMuEll}\label{thm:hfree-fpt-mu-ell}
 Let $h\geq3$ and $H=K_h-e$.
 An instance of \HFEwFD{} given together with a modification-disjoint collection~$\mathcal H$ of induced copies of~$H$ in the input graph can be decided in $O(h^{4(\mu+\ell)}\cdot h^2n^{h})$ time, where $\mu=k-|\mathcal H|$.
\end{restatable}%
\looseness=-1 The algorithm adapts the one of \cref{thm:fpt-mu-ell}.
It takes adding the missing pairs of all members of~$\mathcal H$ as the default and branches on an induced copy~$R$ of~$H$ in the resulting graph: it either adds the missing pair of~$R$, which uses up excess, or deletes an edge of~$R$, or, if that edge is the missing pair of a member of~$\mathcal H$, deletes one of the edges of that member, which uses up a deletion.
\toappendix{
  \sv{\hfreeFptMuEll*}
\begin{proof}[Proof sketch]
 We adapt the algorithm of \cref{thm:fpt-mu-ell}.
 Every solution has to toggle a pair of each of the pairwise disjoint sets $\mods(Q)$, $Q\in\mathcal H$, and hence has at least $|\mathcal H|$ edits; adding all missing pairs, which yields $G^{\mathcal H}\coloneqq(V(G),E(G)\cup\{\miss(Q)\mid Q\in\mathcal H\})$, attains this bound without any deletion.
 Thus we answer ``no'' if $\mu<0$ or $\ell<0$, and ``yes'' if $G^{\mathcal H}$ is $H$-free.

 Otherwise, we find an induced copy~$R$ of~$H$ in~$G^{\mathcal H}$ in $O(h^2n^h)$ time, let $\{a,b\}$ be its missing pair, and branch on how a solution~$F$ destroys~$R$ as follows.
 Since $G\triangle F$ is $H$-free, the vertices of~$R$ do not induce the same copy of~$H$ in it, so $F$ adds $\{a,b\}$ or some edge~$e$ of~$R$ is not an edge of $G\triangle F$.
 Each of the following branches records one edit~$f$, decreases~$k$ by one, and keeps the collection $\mathcal H-f\coloneqq\{Q\in\mathcal H\mid f\notin\mods(Q)\}$:
 (i)~add $f=\{a,b\}$; since $\{a,b\}\notin E(G^{\mathcal H})$ and $\mods(Q)\subseteq E(G^{\mathcal H})$ for every $Q\in\mathcal H$, we have $\mathcal H-f=\mathcal H$, so this edit decreases~$\mu$;
 (ii)~for each edge~$e$ of~$R$ with $e\in E(G)$, delete $f=e$;
 (iii)~for each edge~$e$ of~$R$ with $e\notin E(G)$, we have $e=\miss(Q_e)$ for a unique $Q_e\in\mathcal H$; a solution that does not add~$e$ still has to destroy the induced copy~$Q_e$ of~$H$ in~$G$, and all pairs in $\mods(Q_e)\setminus\{e\}$ are edges of~$G$, so for each of these $\binom{h}{2}-1$ edges~$f$ we delete~$f$.
 In (ii) and (iii), $\ell$ decreases by one and $|\mathcal H|$ by at most one, by modification disjointness, so $\mu+\ell$ decreases by at least one in each of the at most $1+(\binom{h}{2}-1)^2\leq h^4$ branches.
 As in \cref{lem:brule-well-formed,lem:brule-safe}, the members of $\mathcal H-f$ remain induced copies of~$H$ after toggling~$f$; a solution~$F'$ of a produced instance yields the solution $F'\triangle\{f\}$ of the original one; and, conversely, the case distinction above shows that every solution~$F$ contains the pair~$f$ of some branch, and $F\setminus\{f\}$ solves the corresponding instance.
 This yields the claimed search-tree size and running time; for $h=3$, this is the five-way branching of \cref{brule:destroy-p3}, before its improvement via \cref{lem:p3-with-original-edge}.
\end{proof}
}

\subsection{\text{CEwFD} with Modification-disjoint Star Packings}
\looseness=-1
We consider a parameterization above a guarantee given by a packing of other induced subgraphs.
For a star $S$ on at least $3$ vertices, let $\lb(S) \coloneqq |V(S)|-2$.
On the one hand, note that any star $S$ can be turned into a cluster graph by modifying $\lb(S)$ edges, namely by deleting this many edges.
On the other hand, any edge modification of $S$ to a cluster graph, requires at least $\lb(S)$ modifications.

Let $\mathcal{S}$ be a collection of modification-disjoint induced stars in $G$, each with at least 3 vertices, and let $\sigma\coloneqq k - \sum_{S \in \mathcal{S}}(\lb(S))$.

\begin{theorem}\label{thm:fpt-sigma-ell}
\textsc{CEwFD} is FPT with respect to the combination of $\ell$ and $\sigma$.
\end{theorem}

We first bound the size and the number of large stars in~$\mathcal S$.

\begin{observation}\label{obs:star-size}
Let $(G,k,\ell,\mathcal S)$ be a packed instance. If there is a star $S\in\mathcal S$ with more than $s\coloneqq\ell+2\sqrt{\sigma}+3$ vertices, then there is no solution. 
\end{observation}

\begin{proof}
    Indeed, even if we delete $\ell$ edges from~$S$, there are still at least $2\sqrt{\sigma} + 4$ vertices left, which end up in the same connected component.
 Hence, the corresponding clique needs at least $\binom{2\sqrt{\sigma}+4}2 = \frac12(2\sqrt{\sigma}+4)(2\sqrt{\sigma}+3)=2\sigma + 7\sqrt{\sigma}+6$ edges.
 But at most $2\sqrt{\sigma}+3$ of them are present and at most another $2\sqrt{\sigma}+2$ was counted for in the lower bound $\lb(S)$.
 So, we would have to add at least $2\sigma+1$ extra edges, which exceeds $\sigma$.
\end{proof}

\begin{observation}\label{obs:few-large-stars}
    Let $(G,k,\ell,\mathcal S)$ be a packed instance. If there are more than $\ell+\sigma$ stars of size at least 4 in $\mathcal S$, then there is no solution. 
\end{observation}

\begin{proof}
    Each star~$S$ with $s \ge 4$ vertices requires at least one edge deletion or at least 
    $\binom{s}{2} -(s-1) = \frac{s}{2}(s-1) - (s-1) = \frac{s-2}{2}(s-1)$
    edge additions, which is at least $s-1 > \lb(S)$ since $s \ge 4$.
 Thus, if there are more than $\ell +\sigma$ such stars, then any solution requires more than~$k$ edge modifications or more than $\ell$ deletions.
\end{proof}

\begin{proof}[Proof of \cref{thm:fpt-sigma-ell}]
 By \cref{obs:star-size,obs:few-large-stars},
 we know that each star in~$\mathcal S$ contains at most
 $\ell+2\sqrt{\sigma}+3$ vertices and there are at most $\ell+\sigma$ stars with more than 3 vertices (non-$P_3$s).
 Therefore, the total number of vertices contained in non-$P_3$ stars in~$\mathcal S$ is in $O((\ell+\sigma)^2)$ and we can guess the cluster graph on these vertices (the number of possible partitions of these vertices is clearly upper-bounded by some function of~$\ell+\sigma$), remove the non-$P_3$ stars from~$\mathcal S$, and compute the remaining budget $k'\le k$ and $\ell'\le \ell$ of edge modifications. In the remaining packed instance $(G',k',\ell',\mathcal S')$ the packing~$\mathcal S'$ contains only~$P_3$s and we have an above-$P_3$ parameter $\mu\coloneqq k'-|\mathcal S'|\le k - \sum_{S\in\mathcal S\setminus\mathcal S'}\lb(S)-|\mathcal S'|=\sigma$. We can therefore solve this instance in FPT-time wrt.\ $\ell+\sigma$ using the algorithm from \cref{thm:fpt-mu-ell} for the parameter~$\ell+\mu$.
\end{proof}


\section{Kernelization wrt. Deletions \boldmath$\ell$ Combined with Above-guarantees}
\toappendix{\sv{\section{Missing Details for the Kernelization Results}\label{app:kernel}}}

\subsection{Cluster Editing with Few Deletions and Above Packing of \boldmath$P_3$s}

\begin{theorem}\label{thm:kernel-mu-ell}
\textsc{CEwFD} has a kernel with $O((\ell+\mu)^2)$ vertices and $O((\ell+\mu)^3)$ edges. 
\end{theorem}

Let $\mathcal{P}$ be the packing associated with $\mu$, graph $G$ be the input graph, and recall that $G^{\mathcal{P}}$ is obtained from $G$ by adding all missing edges for $P_3$s in $\mathcal{P}$.
To prove the theorem, we start with the following crucial lemma
which connects solutions for $G$ to solutions of $G^\mathcal P$.

\begin{lemma}\label{lem:ker_pathpack_small_solution}
 Let $(G,k,\ell,\mathcal P)$ be a packed instance.
 For every solution $F$ for $(G,k,\ell)$, we have $|E(G^{\mathcal P})\triangle(E(G) \triangle F)| = |F \triangle \miss(\mathcal P)| \le 2\ell + \mu$.
 That is, there is an edge-modification set $F'\subseteq \binom{V}{2}$ of size at most $2\ell + \mu$ such that $G^{\mathcal P}\triangle F'$ is a cluster graph.
\end{lemma}

\begin{proof}
 Let $F$ be a solution for $(G,k,\ell)$, that is, $G \triangle F$ is a cluster graph, $|F| \le k$, and $|F \cap E(G)| \le \ell$.
 Let~$F^+\coloneqq F \setminus E(G)$ be the added edges and~$F^- \coloneqq F \cap E(G)$ be the deleted edges of~$F$.
 We let $F' \coloneqq F \triangle \miss(\mathcal P)$.
 For every deleted edge $e\in F^-$, also $F'$ simply deletes $e$ and if~$e$ is in a~$P_3$ of~$\mathcal P$, then~$F'$ also deletes from~$G^\mathcal P$ the corresponding edge in~$\miss(\mathcal P)$ of the $P_3$.
 So, as $|F^-| \le \ell$, $F'$ deletes at most $2\ell$ edges from~$G^\mathcal P$.
 Moreover, $F'$ adds all edges~$e\in (F^+\setminus \miss(\mathcal P))$.

 As for every $P_3$ $Q$ of~$\mathcal P$, we have $\mods(Q) \cap F^- \neq \emptyset$ or $\miss(Q) \in F^+$, we have $|\mathcal P| \le |F^-|+|F^+\cap \miss(\mathcal P)|$.
 Therefore, $|F^+\setminus \miss(\mathcal P)| \le |F| - |\mathcal P| \le k - |\mathcal P| = \mu$.
\end{proof}

If $k<|\mathcal P|$ or $\ell<0$, then we return no (or a trivial no-instance). 
The kernelization proceeds by computing the graph $G^{\mathcal P}$ and then first applying the following reduction rule, which simply removes isolated cliques.

\begin{reductionrule}\label{rr:isolated_clique}
 If $(G,k,\ell,\mathcal P)$ is a packed instance and $C$ is an isolated clique in $G^{\mathcal P}$, then remove $C$ from $G$, decrease $k$ by $|\binom{C}{2} \setminus E(G)|$ and remove from $\mathcal{P}$ every $P_3$ that contains a vertex of $C$.
\end{reductionrule}

\begin{restatable}[\appmark]{lemma}{isolatedCliqueSafe}\label{lem:isolated_clique_safe}
 \Cref{rr:isolated_clique} is safe, that is, the resulting instance is a yes-instance if and only if the original instance is.
 Furthermore, it does not increase the parameter $\mu+\ell$.
\end{restatable}
\toappendix{
  \sv{\isolatedCliqueSafe*}
\begin{proof}
 Denote the new instance $(G', k', \ell, \mathcal P')$.
 Note that $C$ must form a connected component of $G$, as $G^{\mathcal P}$ contains all edges of $G$ and if two vertices are connected by an edge $\{u,v\} \in \miss{\mathcal{P}}$, then $u$ and $v$ are connected by the corresponding $P_3$ of $\mathcal P$ in $G$. 
 Also each $P_3$ in $\mathcal{P}$ is either fully contained in $C$ or disjoint with $C$, as otherwise there would be an edge between the vertices in $C$ and vertices outside $C$ in $G^{\mathcal P}$.
 Let us denote $\mathcal{P}_C \coloneqq \mathcal{P} \setminus \mathcal{P}'$ the set of those $P_3$'s from $\mathcal{P}$ that are fully contained in~$C$.
 
 If $F'$ is a solution for the new instance, then $F \coloneqq F' \cup (\binom{C}{2} \setminus E(G))$ is a solution to the original instance.
 Indeed, $G'\triangle F'$ is a cluster graph, $G[C] \triangle (\binom{C}{2} \setminus E(G))$ is a clique, $F$ contains only edge deletions already contained in $F'$ and as $|F'|\le k'$, we also have $|F| = |F'|+|\binom{C}{2} \setminus E(G)| \le k' + |\binom{C}{2} \setminus E(G)|=k$.
 
 Now suppose that $F$ is a solution for $(G,k,\ell)$.
 For each $P_3$ $Q$ in $\mathcal{P}$, we know that $F \cap \mods(Q) \neq \emptyset$ and that these sets are disjoint for different $Q$'s.
 Therefore, $|F \cap \{\mods(Q) \mid  Q \in \mathcal{P}_C\}| \ge |\mathcal{P}_C| = |\binom{C}{2} \setminus E(G)|$.
 Let $F' \coloneqq F \cap \binom{V(G) \setminus C}{2}$.
 By the above argument, we have $|F'| \le |F| \setminus |F \cap \{\mods(Q) \mid  Q \in \mathcal{P}_C\}| \le k -|\binom{C}{2} \setminus E(G)| = k'$.
 Also $|F' \cap E(G)| \le |F \cap E(G)| \le \ell$.
 Finally, $G'\triangle F' = (G \triangle F) \setminus C$ is a cluster graph.
 Therefore, $F'$ is a solution for $(G',k', \ell)$. 
 
 As for the parameter, $\ell$ remains unchanged, while $\mu' \coloneqq k' \setminus |\mathcal{P}'|= k - |\binom{C}{2} \setminus E(G)| -(|\mathcal{P}| - |\mathcal{P}_C|) = k - |\mathcal{P}|=\mu$.
\end{proof}
}

Now we use the algorithm of Cao et al.~\cite{CaoCL0NV24} on $G^{\mathcal P}$ to obtain a $1.437$-approximate cluster editing solution~$A$ in polynomial time.
Let $\alpha \coloneqq 1.437$ be the approximation factor.
By \Cref{lem:ker_pathpack_small_solution}, we know that if $(G,k,\ell)$ is a yes-instance, then there is a cluster editing solution for $G^{\mathcal P}$ of size at most $2\ell + \mu$.
Therefore, if $A$ is of size more than $\alpha(2\ell + \mu)$, we can immediately answer no (or return a trivial no-instance).
Otherwise, let $G''\coloneqq G^{\mathcal P}\triangle A$.
Note that $G''$ is a cluster graph.
The following rule essentially states that large cliques in $G''$ are contained in every optimal solution and can thus be commited.

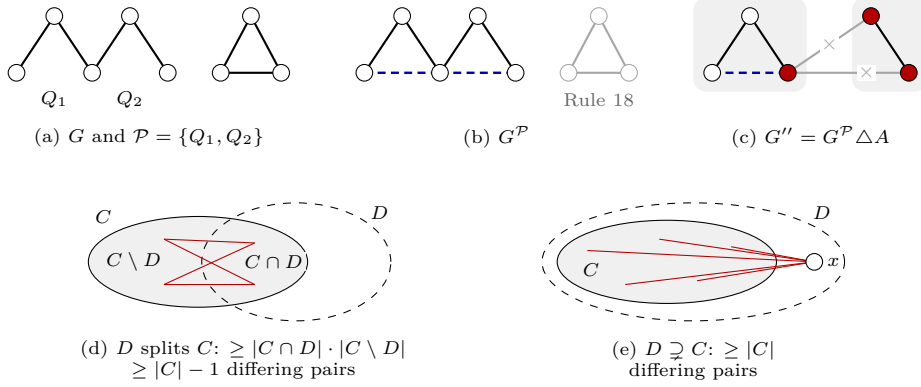
\begin{figure}[t]
\centering
\begin{tikzpicture}[
  vtx/.style={draw, circle, fill=white, minimum size=2.2mm, inner sep=0pt},
  xvtx/.style={draw, circle, fill=red!70!black, minimum size=2.2mm, inner sep=0pt},
  gedge/.style={black, line width=.8pt},
  added/.style={blue!70!black, densely dashed, line width=.9pt},
  gone/.style={black!35, line width=.8pt},
  xmark/.style={midway, inner sep=.4pt, fill=white, font=\scriptsize},
  lbl/.style={font=\scriptsize},
  clus/.style={draw=black!30, fill=black!6, rounded corners=6pt}
]
  \newcommand{\kpos}{%
    \coordinate (a) at (0,0);   \coordinate (b) at (.5,.75); \coordinate (c) at (1,0);
    \coordinate (d) at (1.5,.75); \coordinate (e) at (2,0);
    \coordinate (f) at (2.7,0); \coordinate (g) at (3.1,.75); \coordinate (h) at (3.5,0);}
  \newcommand{\knodes}[1]{\foreach \v in {#1} \node[vtx] at (\v) {};}

  \begin{scope}
    \kpos
    \draw[gedge] (a) -- (b) -- (c) -- (d) -- (e) (f) -- (g) -- (h) -- (f);
    \knodes{a,b,c,d,e,f,g,h}
    \node[lbl] at (.5,-.35) {$Q_1$};  \node[lbl] at (1.5,-.35) {$Q_2$};
    \node[lbl] at (1.75,-.85) {(a) $G$ and $\mathcal P=\{Q_1,Q_2\}$};
  \end{scope}

  \begin{scope}[shift={(4.6,0)}]
    \kpos
    \draw[gedge] (a) -- (b) -- (c) -- (d) -- (e);
    \draw[added] (a) -- (c) (c) -- (e);
    \draw[gone] (f) -- (g) -- (h) -- (f);
    \knodes{a,b,c,d,e}
    \foreach \v in {f,g,h} \node[vtx, draw=black!35] at (\v) {};
    \node[lbl, black!60] at (3.1,-.35) {Rule~\ref*{rr:isolated_clique}};
    \node[lbl] at (1.75,-.85) {(b) $G^{\mathcal P}$};
  \end{scope}

  \begin{scope}[shift={(9.2,0)}]
    \kpos
    \coordinate (d) at (2.1,.75); \coordinate (e) at (2.6,0);
    \fill[black!6, rounded corners=5pt] (-.25,-.25) rectangle (1.25,1.0);
    \fill[black!6, rounded corners=5pt] (1.85,-.25) rectangle (2.85,1.0);
    \draw[gedge] (a) -- (b) -- (c) (d) -- (e);
    \draw[added] (a) -- (c);
    \draw[gone] (c) -- (d) node[xmark] {$\times$};
    \draw[gone] (c) -- (e) node[xmark, pos=.65] {$\times$};
    \knodes{a,b}
    \foreach \v in {c,d,e} \node[xvtx] at (\v) {};
    \node[lbl] at (1.3,-.85) {(c) $G''=G^{\mathcal P}\triangle A$};
  \end{scope}

  \begin{scope}[shift={(2.4,-2.5)}]
    \fill[black!6] (0,0) ellipse (1.45 and .6);
    \draw (0,0) ellipse (1.45 and .6);
    \draw[dashed] (1.3,0) ellipse (1.25 and .78);
    \foreach \ya/\yb in {.3/.25, .3/-.3, -.3/.25, -.3/-.3}
      \draw[red!70!black] (-.45,\ya) -- (.75,\yb);
    \node[lbl, fill=black!6, inner sep=1pt] at (-.85,0) {$C\setminus D$};
    \node[lbl, fill=black!6, inner sep=1pt] at (1.0,0) {$C\cap D$};
    \node[lbl] at (-1.25,.6) {$C$};  \node[lbl] at (2.4,.65) {$D$};
    \node[lbl, align=center] at (.6,-1.3) {(d) $D$ splits $C$: $\geq|C\cap D|\cdot|C\setminus D|$\\$\geq|C|-1$ differing pairs};
  \end{scope}

  \begin{scope}[shift={(8.6,-2.5)}]
    \draw[dashed] (.35,0) ellipse (2.0 and .78);
    \fill[black!6] (0,0) ellipse (1.45 and .55);
    \draw (0,0) ellipse (1.45 and .55);
    \node[lbl] at (-1,-.1) {$C$};  \node[lbl] at (2.05,.65) {$D$};
    \coordinate (x) at (1.95,0);
    \foreach \pt in {(-1.05,.15), (-.55,-.3), (-.1,.3), (.4,-.25), (.85,.2)} \draw[red!70!black] (x) -- \pt;
    \node[xvtx, fill=white] at (x) {};  \node[lbl, right] at (2.0,0) {$x$};
    \node[lbl, align=center] at (.35,-1.3) {(e) $D\supsetneq C$: $\geq|C|$\\differing pairs};
  \end{scope}
\end{tikzpicture}
\caption{Top: the kernelization of \cref{thm:kernel-mu-ell} on a small example.
(a)~A graph~$G$ with a modification-disjoint packing of two induced $P_3$'s sharing the vertex in the middle.
(b)~The graph~$G^{\mathcal P}$ adds their missing pairs (dashed); the isolated clique on the right is removed by \cref{rr:isolated_clique}.
(c)~The cluster graph~$G''$ obtained from an approximate solution~$A$ that removes the two crossed-out pairs; every cluster of~$G''$ contains an endpoint of~$A$ (filled), which bounds the number of clusters.
Bottom: why \cref{rr:large_clique} may commit a clique~$C$ of~$G''$ with at least $(\alpha+1)(2\ell+\mu)+2$ vertices.
If a solution's cluster~$D$ splits~$C$~(d), or contains~$C$ and a further vertex~$x$~(e), then the solution differs from~$G''$ in at least $|C|-1$ vertex pairs (red), more than the $(\alpha+1)(2\ell+\mu)$ allowed by \cref{eq:edge_diff}.}
\label{fig:kernel}
\end{figure}

\begin{reductionrule}\label{rr:large_clique} 
Let $A$ be a cluster-editing set for $G^{\mathcal P}$ of size at most $\alpha(2\ell + \mu)$ and let $G''\coloneqq G^{\mathcal P}\triangle A$.
 Suppose that $C$ is a clique in $G''$ of size at least $(\alpha+1)(2\ell+\mu)+2$ and let $F_C\coloneqq(E(G'') \triangle E(G)) \setminus \binom{V(G) \setminus C}{2}$.
 If $|F_C| > k$ or $|F_C \cap E(G)| > \ell$, then answer no (or return a trivial no-instance).
 Otherwise, let $G' \coloneqq G \setminus C$, $k' \coloneqq k - |F_C|$, $\ell' \coloneqq \ell - |F_C \cap E(G)|$, and $\mathcal{P}' \coloneqq \mathcal{P} \setminus \{ Q \mid F_C \cap \mods(Q) \neq \emptyset\}$ and return instance $(G', k', \ell', \mathcal{P}')$.
\end{reductionrule}

\begin{lemma}\label{lem:large_clique_safe}
 \Cref{rr:large_clique} is safe.
 Furthermore, it does not increase 
 $\mu+\ell$.
\end{lemma}

\begin{proof}
  \looseness=-1
 We first show that each solution for $(G,k,\ell)$ contains $F_C$.
 Suppose that there is a solution $\widehat{F}$ for $(G,k,\ell)$ that does not contain some edge of $F_C$.
 By \Cref{lem:ker_pathpack_small_solution}, $|\widehat{F} \triangle \miss(\mathcal{P})| \le 2\ell+\mu$.
 Therefore, the two cluster graphs $G \triangle \widehat{F}$ and $G''$ differ in at most $(\alpha+1)(2\ell+\mu)$ edges, since
 \begin{equation}
 |(\widehat{F} \triangle E(G)) \triangle (A \triangle E(G^{\mathcal{P}}))| \le |\widehat{F} \triangle \miss(\mathcal{P})| + |A| \le (2\ell+\mu) + \alpha(2\ell+\mu) = (\alpha+1)(2\ell+\mu). \label{eq:edge_diff}
 \end{equation} 
  As $F_C$ contains a vertex pair which is not in~$\widehat{F}$  and all modifications in $F_C$ are incident to~$C$, this means that in $G \triangle \widehat{F}$, either $C$ is a part of a larger cluster or it is split into several connected components.
  Suppose the latter, that is, there is a cluster $D$ of $G \triangle \widehat{F}$ such that $C \cap D \neq \emptyset$ and $C \nsubseteq D$.
  Then, $G''$ and $G\triangle\widehat{F}$ differ at least in all edges between $C \cap D$ and $C \setminus D$. 
  This is at least $|C \cap D|(|C|-|C \cap D|)$ edges for some $|C \cap D| \in \{1, \ldots, |C|-1\}$.
  As this is a concave function in~$|C\cap D|$, its minimum on the range is achieved either by $|C \cap D| = 1$ or $|C \cap D| =|C|-1$.
  Therefore, the minimum is $|C|-1 = (\alpha+1)(2\ell+\mu)+1$, i.e., the cluster graphs differ by at least $(\alpha+1)(2\ell+\mu)+1$ edges, contradicting \cref{eq:edge_diff}.
  Hence, there is a cluster $D$ of $G \triangle \widehat{F}$ such that $C \subseteq D$.
  If it contains also a vertex $x \in (D \setminus C)$, then $G''$ and $G\triangle\widehat{F}$ would differ in at least the edges between $x$ and $C$, which is at least $|C| \ge (\alpha+1)(2\ell+\mu)+2$ edges, again contradicting \cref{eq:edge_diff}.
  Therefore, every solution for $(G,k,\ell)$ contains $F_C$.
  Thus, if the rule returned ``no'', this is correct.

 Now, assume that $F'$ is a solution for the instance $(G', k', \ell')$. We claim that $F = F' \cup F_C$ is a solution for the original instance.
 Indeed, $G'\triangle F'$ is a cluster graph, $C$ is an isolated clique in $G \triangle F_C$ and therefore $G'' \triangle F = (G'' \triangle F_C) \triangle F'$ is a cluster graph, as $F_C \cap F' = \emptyset$.
 Since $|F'| \le k'$, we have $|F| = |F'| +|F_C| \le k' + |F_C| =k$ and as $|F' \cap E(G)| \le \ell'$, we have $|F \cap E(G)| \le \ell' + |F_C \cap E(G)|= \ell$.
 
  Conversely, let $F$ be a solution for $(G,k,\ell)$. 
  We claim that $F' \coloneqq F \setminus F_C$ is a solution for the instance $(G', k', \ell')$.
  As $F_C \subseteq F$, we have that $|F'| = |F| -|F_C| \le k-|F_C| = k'$ and $|F'\cap E(G)| = |F \cap E(G)| - |F_C \cap E(G)| \le \ell - |F_C \cap E(G)| = \ell'$.
  Furthermore, as~$F$ and~$F'$ only differ in edges incident to $C$, it follows that $G' \triangle F' = (G \triangle F) \setminus C$ is a cluster graph.
  Therefore, $(G', k', \ell')$ is indeed a yes-instance.
  
  Since the set $\mathcal{P}_C \coloneqq \{ Q \mid F_C \cap \mods(Q) \neq \emptyset\}$ is not larger than $F_C$, we have $\mu' \coloneqq k' - |\mathcal{P}'| = k - |F_C| -(|\mathcal{P}| - |\mathcal{P}_C|) \le k - |\mathcal{P}|=\mu$.
  The bound on the parameter follows, as also $\ell' \le \ell$.
\end{proof}

\begin{proof}[Proof of \Cref{thm:kernel-mu-ell}]
The kernelization answers ``no'' if $k<|\mathcal P|$ or $\ell<0$.
Otherwise, it applies \cref{rr:isolated_clique} as long as possible, computes an $\alpha$-approximate solution~$A$ for~$G^{\mathcal P}$, answers ``no'' if $|A|>\alpha(2\ell+\mu)$, and otherwise applies \cref{rr:large_clique} if possible and starts over.
The answers are correct by \cref{lem:ker_pathpack_small_solution}, the rules are safe and do not increase $\mu+\ell$ by \cref{lem:isolated_clique_safe,lem:large_clique_safe}, and each application of a rule removes vertices, so the kernelization runs in polynomial time.

\looseness=-1
It remains to show the size of the kernel. 
It is enough to bound the size of an instance to which no rule applies and for which $|A|\leq\alpha(2\ell+\mu)$, as the recursive calls do not increase the parameter by \Cref{lem:isolated_clique_safe} and \Cref{lem:large_clique_safe}. 
Let $X$ be the set of vertices incident to~$A$.  
Clearly, $|X| \le 2|A|\le 2\alpha(2\ell +\mu)$.
Let again $G''\coloneqq G^{\mathcal P}\Delta A$ be the cluster graph.
If a clique in $G''$ does not contain any vertex of~$X$, then it is also a clique in $G^\mathcal P$ and \Cref{rr:isolated_clique} would apply.
Therefore, each clique of~$G''$ contains a vertex of~$X$, that is, there are at most $2\alpha(2\ell +\mu)$ cliques.
Furthermore, each clique is of size at most $(\alpha+1)(2\ell+\mu)+1$, as otherwise \Cref{rr:large_clique} would apply.
Therefore, $G''$ contains $O((\ell+\mu)^2)$ vertices and $O((\ell+\mu)^3)$ edges.
As $G^{\mathcal{P}}$ and $G''$ only differ in edges of $A$ and $|A| \in O(\ell+\mu)$, also $G^{\mathcal{P}}$ has $O((\ell+\mu)^3)$ edges.
Finally, $G$ has the same vertices as $G''$ and at most as many edges as $G^{\mathcal{P}}$, therefore it has $O((\ell+\mu)^2)$ vertices and $O((\ell+\mu)^3)$ edges.
\end{proof}

\subsection{Cluster Editing with Few Deletions and Above Packing of Stars}

\begin{restatable}[\appmark]{theorem}{kernelSigmaEll}\label{thm:kernel-sigma-ell}
\textsc{CEwFD} has a kernel with $O((\ell+\sigma)^2(\ell^2+\sigma)^2)$ vertices and $O((\ell+\sigma)^3(\ell^2+\sigma)^3)$ edges.

\end{restatable}
\todo{maybe just say polynomial problem kernel?}\todo{F: If the proof is not part of the short version, I would keep it, and otherwise I am fine with dropping it}
\looseness=-1
The kernel follows the one of \cref{thm:kernel-mu-ell}, with $\mathcal P$ being the set of $P_3$'s in a star packing~$\mathcal S$.
By \cref{obs:star-size,obs:few-large-stars}, there are at most $\ell+\sigma$ larger stars in~$\mathcal S$, each with at most $\ell+2\sqrt{\sigma}+3$ vertices, so they span only few vertex pairs.
Together with the argument of \cref{lem:ker_pathpack_small_solution} for the $P_3$'s, this shows that every solution differs from $G^{\mathcal P}$ in at most $(\ell+\sigma)(\ell+2\sqrt{\sigma}+4)^2$ vertex pairs, and the two reduction rules and the size bound carry over with this bound in place of $2\ell+\mu$.
\toappendix{
  \sv{\kernelSigmaEll*}

If $(G,k,\ell)$ is an instance of \textsc{CEwFD} and $\mathcal{S}$ is a collection of modification-disjoint induced stars in $G$, each with at least 3 vertices, then we call $(G,k,\ell,\mathcal S)$ \emph{a star-packed instance}.
We start by proving the following analogue of \Cref{lem:ker_pathpack_small_solution}.
\begin{lemma}\label{lem:ker_starpack_small_solution}
 Let $(G,k,\ell, \mathcal{S})$ be a star-packed instance. 
 Let $\mathcal{P} \subseteq \mathcal{S}$ be the set of $P_3$s in $\mathcal{S}$.
 Let $G^{\mathcal P} \coloneqq G \triangle \miss(\mathcal P)$.
 For every solution $F$ for $(G,k,\ell)$, we have $|E(G^{\mathcal P})\triangle(E(G) \triangle F)| = |F \triangle \miss(\mathcal P)| \le (\ell+\sigma)(\ell+2\sqrt{\sigma}+4)^2$.
 That is, there is an edge-modification set $F'\subseteq \binom{V}{2}$ of size at most $(\ell+\sigma)(\ell+2\sqrt{\sigma}+4)^2$ such that $G^{\mathcal P}\triangle F'$ is a cluster graph.
\end{lemma}

\begin{proof}
 Let $F$ be a solution for $(G,k,\ell)$, that is, $G \triangle F$ is a cluster graph, $|F| \le k$, and $|F \cap E(G)| \le \ell$.
 By \Cref{obs:star-size}, each star in $\mathcal{S}$ has at most $\ell+2\sqrt{\sigma}+3$ vertices and, by \Cref{obs:few-large-stars}, there are at most $\ell+\sigma$ stars on at least 4 vertices in $\mathcal{S}$, i.e., denoting $\widehat{\mathcal{S}} \coloneqq \mathcal{S} \setminus \mathcal{P}$, we have $|\widehat{\mathcal{S}}| \le (\ell+\sigma)$.
 Let $Y \coloneqq \bigcup_{S \in \widehat{\mathcal{S}}} \binom{V(S)}{2}$ be the set of vertex pairs within the large stars.
 By the above bounds, we have $|Y| \le (\ell+\sigma) \binom{\ell+2\sqrt{\sigma}+3}{2} \le (\ell+\sigma)(\ell+2\sqrt{\sigma}+3)^2$.
 We also have $\lb(\widehat{\mathcal{S}}) \coloneqq \sum_{S \in \widehat{\mathcal{S}}}\lb(S) \le (\ell+\sigma)(\ell+2\sqrt{\sigma}+1)$.
 Note that, since the stars in $\mathcal{S}$ are modification-disjoint, we have $Y \cap \miss(\mathcal P) = \emptyset$.
 
 Now we proceed to bound $|(F\setminus Y) \triangle \miss(\mathcal P)|$ in a similar way as in \Cref{lem:ker_pathpack_small_solution}.
 Let~$F^+\coloneqq (F\setminus Y) \setminus E(G)$ be the added edges and~$F^- \coloneqq (F \setminus Y) \cap E(G)$ be the deleted edges of~$F \setminus Y$.
 Note that we clearly have $|F^-| \le |F \cap E(G)| \le \ell$ and also $|F^+|+|F^-| \le |F| \le k$.
 We let $F'' \coloneqq (F\setminus Y)\triangle\miss(\mathcal P)$.
 For every deleted edge $e\in F^-$, also $F''$ simply deletes $e$ and if~$e$ is in a~$P_3$ of~$\mathcal P$, then~$F''$ also deletes from~$G^\mathcal P$ the corresponding edge in~$\miss(\mathcal P)$ of the $P_3$.
 So, as $|F^-| \le \ell$, $F''$ deletes at most $2\ell$ edges from~$G^\mathcal P$.
 Moreover, $F''$ adds all edges~$e\in (F^+\setminus \miss(\mathcal P))$.

 As for every $P_3$ $Q$ of~$\mathcal P$, we have $\mods(Q) \cap F^- \neq \emptyset$ or $\miss(Q) \in F^+$, we have $|\mathcal P| \le |F^-|+|F^+\cap \miss(\mathcal P)|$.
 Therefore, \[|F^+\setminus \miss(\mathcal P)| \le |F| - |\mathcal P| \le k - |\mathcal P| = \sigma+\lb(\widehat{\mathcal{S}}) \le  \sigma+ (\ell+\sigma)(\ell+2\sqrt{\sigma}+1).\]
 Together, we have $|F''| \le 2\ell+\sigma+(\ell+\sigma)(\ell+2\sqrt{\sigma}+1) \le (\ell+\sigma)(\ell+2\sqrt{\sigma}+3)$.
 
 Finally, we let $F' \coloneqq F \triangle \miss(\mathcal P)$.
 Note that, since $Y \cap \miss(\mathcal P) = \emptyset$, we have \[F' = (F \cap Y) \cup ((F \setminus Y) \Delta \miss(\mathcal P)) = (F \cap Y) \cup F''.\]
 Thus, \[|F'| \le |Y| + |F''| \le (\ell+\sigma)(\ell+2\sqrt{\sigma}+3)^2 + (\ell+\sigma)(\ell+2\sqrt{\sigma}+3) \le (\ell+\sigma)(\ell+2\sqrt{\sigma}+4)^2.\qedhere\]
\end{proof}

Now we adapt the two rules from the previous section.

\begin{reductionrule}\label{rr:star_isolated_clique}
 If $(G,k,\ell,\mathcal S)$ is a star-packed instance, $\mathcal{P}$ is the set of $P_3$s in $\mathcal{S}$, $G^{\mathcal P} \coloneqq G \triangle \miss(\mathcal P)$, and $C$ is an isolated clique in $G^{\mathcal P}$, then remove $C$ from $G$, decrease $k$ by $|\binom{C}{2} \setminus E(G)|$ and remove from $\mathcal{S}$ every $P_3$ that contains a vertex of $C$.
\end{reductionrule}

\begin{lemma}\label{lem:star_isolated_clique_safe}
 \Cref{rr:star_isolated_clique} is safe and does not increase the parameter $\ell+\sigma$.
\end{lemma}

The proof of \Cref{lem:star_isolated_clique_safe} is completely analogous to the proof of \Cref{lem:isolated_clique_safe} (note that we have $V(S) \cap C = \emptyset$ for every $S \in \mathcal{S} \setminus \mathcal{P}$) and we therefore omit it.

\begin{reductionrule}\label{rr:star_large_clique} 
Let $(G,k,\ell,\mathcal S)$ be a star-packed instance, $\mathcal{P}$ be the set of $P_3$s in $\mathcal{S}$, and $G^{\mathcal P} \coloneqq G \triangle \miss(\mathcal P)$.
Let $A$ be a cluster-editing set for $G^{\mathcal P}$ of size at most $\alpha(\ell+\sigma)(\ell+2\sqrt{\sigma}+4)^2$ for some constant $\alpha$ and let $G''\coloneqq G^{\mathcal P}\triangle A$.
 Suppose that $C$ is a clique in $G''$ of size at least $(\alpha+1)(\ell+\sigma)(\ell+2\sqrt{\sigma}+4)^2+2$ and let $F_C\coloneqq(E(G'') \triangle E(G)) \setminus \binom{V(G)-C}{2}$.
 If $|F_C| > k$ or $|F_C \cap E(G)| > \ell$, then answer no.
 Otherwise, let $G' \coloneqq G \setminus C$, $k' \coloneqq k - |F_C|$, $\ell' \coloneqq \ell - |F_C \cap E(G)|$, and $\mathcal{S}' \coloneqq \{S \setminus C \mid |V(S) \setminus C| \ge 3 \wedge S \setminus C \text{ is a star}\}$ and return instance $(G', k', \ell', \mathcal{S}')$.
\end{reductionrule}

\begin{lemma}\label{lem:star_large_clique_safe}
 \Cref{rr:star_large_clique} is safe. It also does not increase the parameter $\ell+\sigma$.
\end{lemma}

\begin{proof}
 The proof that any solution for $(G,k,\ell)$ contains $F_C$ is completely analogous to the proof in \Cref{lem:large_clique_safe}, using \Cref{lem:ker_starpack_small_solution} instead of \Cref{lem:ker_pathpack_small_solution}.
 Based on that, the proof of safeness is exactly the same.
 
 To show that the parameter does not increase, we need to bound $\lb(\mathcal{S}) - \lb(\mathcal{S}')$ by $|F_C|$.
 We show that for every $S \in \mathcal{S}$, the difference related to the star $S$ can be bounded by $|F_C \cap \binom{V(S)}{2}|$.
 This is sufficient, as the stars are modification disjoint.
 \begin{itemize}
 \item If $(S \setminus C) \notin \mathcal{S}'$, then we distinguish two cases.
 \begin{itemize}
 \item If $V(S) \setminus C$ contains the center $c$ of $S$, then $|V(S) \setminus C| \le 2$. 
 All edges from $c$ to $|V(S) \cap C| \ge |V(S)|-2$ vertices of $S$ in $V(S) \cap C$ were deleted by $F_C$. 
 This is at least $\lb(S)$ edges.
 \item If the center $c$ of $S$ is in $C$, then all edges from $c$ to vertices in $V(S) \setminus C$ must be deleted. 
 Furthermore, if $|V(S) \cap C| \ge 3$, then at least $\lb(S[V(S) \cap C]) = |V(S) \cap C|-2$ modifications must be done by $F_C$ in $S[V(S) \cap C]$.
 Together this is at least $|V(S) \setminus C|+ |V(S) \cap C|-2 = \lb(S)$ edges.
 \end{itemize}
 \item If $(S \setminus C) \in \mathcal{S}'$, then all the edges between center of the star (which is in $V(S) \setminus C$) and vertices in $V(S) \cap C$ must be deleted. 
 This is at least $|V(S)\cap C| = |V(S)|-2-(|V(S) \setminus C|-2) =\lb(S) - \lb(S \setminus C)$ edges.
 \end{itemize}
 Summing over all $S \in \mathcal{S}$, we get $\lb(\mathcal{S}) - \lb (\mathcal{S}') \le |F_C|$, i.e., $\lb(\mathcal{S}) \le  |F_C| + \lb (\mathcal{S}')$.
 Therefore, we have $\sigma' \coloneqq k' - \lb(\mathcal{S}') = k - |F_C| -\lb(\mathcal{S}') \le k -\lb(\mathcal{S}')=\sigma$.
 The bound on the parameter follows, as also $\ell' \le \ell$.
\end{proof}

\begin{proof}[Proof of \Cref{thm:kernel-sigma-ell}]
The algorithm proceeds almost the same as the one of \Cref{thm:kernel-mu-ell}, just using the bound from \Cref{lem:ker_starpack_small_solution} and \Cref{rr:star_isolated_clique,rr:star_large_clique}. 
Thus, its correctness can be also proven completely analogously as for the aforementioned algorithm.

It remains to update the bound on the size of the kernel. 
It is enough to bound the size of the instance output in the last step, that is, if there is an approximate solution~$A$ of size at most $\alpha(\ell+\sigma)(\ell+2\sqrt{\sigma}+4)^2$, but none of the reduction rules applies.
Note that $(\ell+\sqrt{\sigma})^2 \in O(\ell^2+\sigma)$.
Let $X$ be the set of vertices incident to~$A$, we again have $|X| \le 2|A|\le 2\alpha(\ell+\sigma)(\ell+2\sqrt{\sigma}+4)^2$.
Let again $G''\coloneqq G^{\mathcal P}\Delta A$ be the cluster graph.
If a clique in $G''$ does not contain any vertex of~$X$, then it is also a clique in $G^\mathcal P$ and \Cref{rr:star_isolated_clique} would apply.
Therefore, each clique of~$G''$ contains a vertex of~$X$, that is, there are at most $2\alpha(\ell+\sigma)(\ell+2\sqrt{\sigma}+4)^2$ cliques.
Furthermore, each clique is of size at most $(\alpha+1)(\ell+\sigma)(\ell+2\sqrt{\sigma}+4)^2+1$, as otherwise \Cref{rr:star_large_clique} would apply.
Therefore, $G''$ contains $O((\ell+\sigma)^2(\ell+\sqrt{\sigma})^4)=O((\ell+\sigma)^2(\ell^2+\sigma)^2)$ vertices and $O((\ell+\sigma)^3(\ell+\sqrt{\sigma})^6)=O((\ell+\sigma)^3(\ell^2+\sigma)^3)$ edges.
As $G^{\mathcal{P}}$ and $G''$ only differ in edges of $A$ and $|A| \in O((\ell+\sigma)(\ell+\sqrt{\sigma})^2)=O((\ell+\sigma)(\ell^2+\sigma))$, also $G^{\mathcal{P}}$ has $O((\ell+\sigma)^3(\ell^2+\sigma)^3)$ edges.
Finally, $G$ has the same vertices as $G''$ and at most as many edges as $G^{\mathcal{P}}$, therefore it has $O((\ell+\sigma)^2(\ell^2+\sigma)^2)$ vertices and $O((\ell+\sigma)^3(\ell^2+\sigma)^3)$ edges.
\end{proof}
}

\section{Discussion}\label{sec:discussion}

We answered the question of Komusiewicz and Niedermeier~\cite{KomusiewiczN12} in the negative: \CEwFD{} is W[1]-hard with respect to the number~$\ell$ of deletions, even combined with the pathwidth and the feedback edge number.
In contrast, combining~$\ell$ with the excess~$\mu$ above a modification-disjoint $P_3$-packing yields an $O(4^{\mu+\ell}\cdot(n+m+\mu))$-time algorithm and a polynomial problem kernel, although each of the two parameters alone is intractable. 
The FPT-algorithm and the polynomial kernelization also extend to lower bounds by modification-disjoint star packings.
Moreover, most of our results also transfer to $H$-free editing for $H=K_\eta-e$.

Several questions remain open:
(1) Does \textsc{$H$-Free Editing with Few Deletions} admit a polynomial kernel for $\mu+\ell$?
(2) Can $\ell$ be replaced by smaller parameters, and how far below~$\ell$ can we go while keeping fixed-parameter tractability together with~$\mu$?
(3) Which other above-guarantee parameters yield tractability together with~$\ell$? 
Potential candidates are half-integral $P_3$-packings, i.e., where we consider a packing $\mathcal{P}$ of induced $P_3$s such that each vertex pair is in up to two $P_3$ and we parameterize by $k - 2|\mathcal{P}|$, or parameterizing above the LP relaxation of the corresponding $P_3$-packing ILP.
(4) 
One might also study \textsc{Cluster Editing} with a bounded number~$q$ of additions. 
Clearly, one obtains NP-hardness for~$q=0$ due to the NP-hardness of \textsc{Cluster Deletion}.
It is, however, open whether~$q+\mu$ is FPT.
(5) Which reduction rules for \CE{} transfer to \CEwFD{}? Candidates are generic rules that apply to every solution.

\subparagraph*{Declaration on the use of AI.}
Discussions with AI were utilized while obtaining the hardness results.
ChatGPT suggested a convoluted W[1]-hardness proof for \CEwFD{} parameterized by the number~$\ell$ of deletions; the authors subsequently constructed the present more elegant and comprehensible version, for which an initial writeup was provided by Claude Opus 5.5.
Discussions with Fable 5 and 5.1 were utilized to generalize it to \textsc{$H$-Free Editing with Few Deletions} (\Cref{thm:hfree-hardness}).
Claude Opus 5.5 was also used when editing the paper and to write the code for the experiments.
All AI-generated artifacts were checked, corrected, and rewritten by the authors, who take full responsibility for the content of this paper.

\clearpage

\appendix

\appendixText

\section{Experiments}\label{sec:exp}
\begin{table}
  \caption{Absolute magnitudes of the parameters over the 109 instances solved to proven optimality (of 189 attempted; the remaining 80 exceeded the time limit).}
  \label{tab:magnitudes}
  \centering
  \begin{tabular}{lrrrrr}
    \toprule
    & min & 25\% & median & 75\% & max \\
    \midrule
    $k$ & 3 & 128 & 322 & 672 & 1835 \\
    $\ell$ & 0 & 70 & 206 & 374 & 1564 \\
    $\mu$ & 0 & 0 & 1 & 5 & 44 \\
    $\sigma$ & 0 & 0 & 0 & 4 & 44 \\
    $\ell+\mu$ & 0 & 75 & 212 & 374 & 1564 \\
    $\ell+\sigma$ & 0 & 75 & 211 & 374 & 1564 \\
    $\mu_{vd}$ & 1 & 105 & 276 & 646 & 1768 \\
    \bottomrule
  \end{tabular}
\end{table}

\begin{table}
  \caption{Per data set, median values over the instances solved to proven optimality.  The packing columns give the lower bound as a fraction of~$k$: a value of~$1$ means the bound is tight.}
  \label{tab:families}
  \centering
  \begin{tabular}{lrrrrrrr}
    \toprule
    data set & \# & $n$ & $k$ & $\ell/k$ & $|\mathcal P|/k$ & $(\ell+\mu)/k$ & $(\ell+\sigma)/k$ \\
    \midrule
    Jena & 55 & 113--347 & 369 & 0.67 & 1.00 & 0.68 & 0.68 \\
    action sequences & 43 & 10--180 & 281 & 0.86 & 1.00 & 0.89 & 0.89 \\
    strings & 6 & 250--250 & 506 & 0.50 & 1.00 & 0.50 & 0.50 \\
    genes & 4 & 232--232 & 426 & 0.48 & 1.00 & 0.48 & 0.48 \\
    colours & 1 & 330--330 & 200 & 0.62 & 0.96 & 0.67 & 0.67 \\
    \bottomrule
  \end{tabular}
\end{table}

With our small-scale experimental study we wanted to understand how the parameters under consideration relate to the standard parameter $k$, the number of allowed edits.
To this end, we took the graph data from the Parameterized Algorithms and Computational Experiments Challenge 2021~\cite{KellerhalsKNZ21} about exact algorithms for \CE.
We used the unweighted graphs of the exact track, as published in the
challenge's instance repository~\cite{PACEinstances21}.
For each of the instances we then computed:
\begin{itemize}
\item A modification-disjoint $P_3$-packing computed as in the KaPoCE
  solver~\cite{BlasiusFGHHSWW22}, which in turn follows Gottesb\"uren et al.~\cite{GottesburenHSSW20}: We enumerate induced $P_3$s and take them as vertices in a conflict graph where the edges model shared modifications.
  Then we greedily compute an independent set in the conflict graph and then improve it by $2$-improvement local search.
\item A vertex-disjoint $P_3$-packing, computed in the same way except that two $P_3$s are already in conflict when they share a single vertex.
  Note that such a packing consists of at most $\lfloor n/3\rfloor$ many $P_3$s, which caps the lower bound it can give.
\item A modification-disjoint packing of induced stars, obtained by growing each
  $P_3$ of the modification-disjoint packing above into a maximal induced star,
  using only vertex pairs that that packing leaves free.
\item A cluster editing set that among all those minimizing the number $k$ of edits also minimizes the number of deletions. To do this, we  first compute a minimum-edit solution via an ILP over variables that model whether pairs of vertices are merged and row-generated constraints that ensure the partition property. Afterwards, we use a modification of this ILP that fixes the edits to be minimum-possible and minimizes the number of deletions.
\end{itemize}
In the results, $\mu$ is $k$ minus the size of the computed modification-disjoint $P_3$-packing, $\mu_{vd}$ is $k$ minus the size of the computed vertex-disjoint $P_3$-packing, and $\sigma$ is $k$ minus the lower bound given by the computed star packing.

We ran the experiments on an AMD Ryzen 5 5600X (6 cores, 12 threads) with 16 GB RAM under Ubuntu 24.04.5 LTS.
ILPs were solved with SCIP 10.0.2 (SoPlex 8.0.2) driven from Python 3.12.3. Four instances were processed in parallel, each SCIP call capped at 1200 MB.
Per-instance limits were 300\,s for the editing ILP and 120\,s for the packing bounds together.

Of the 189 instances, both stages of the editing ILP finished within the time
limit on 109, and these are the instances we report on.
\Cref{tab:magnitudes} reports the absolute magnitudes of the parameters and
\cref{tab:families} the solved instances and medians per source data set.

Compared to $k$, the deletion parameter $\ell$ was usually half to four fifths as
large (quartiles $0.50$ and $0.83$, median $0.73$).  The above-guarantee
parameters, by contrast, are nearly always negligible: $\mu$ has median~$1$ and
is zero on 54 of the 109 instances, with maximum~$44$, so the packing bound alone
certifies optimality of $k$ on half of them.  This is consistent with the
effectiveness that the PACE participants report for packing-based lower
bounds~\cite{BlasiusFGHHSWW22}.  Because $\mu$ is so small, $\ell+\mu$ is
essentially~$\ell$: its median is $0.74\,k$, so the combined parameter is smaller
than $k$ on 103 of the 109 instances, but only by about a quarter, and it exceeds
$k$ on two.  Growing the $P_3$s into stars improves the bound on 29 instances and
reduces the total excess from $490$ to $449$, which leaves the picture unchanged:
with $\mu$ and $\sigma$ this close to zero, there is almost nothing for a
stronger lower bound to recover.

Requiring the packed $P_3$s to be vertex-disjoint instead changes the picture
completely.  Such a packing captured only a tenth of $k$ in the median, so
$\mu_{vd}$ has median~$276$, amounts to $90\%$ of $k$ in the median, and is
nonzero on every one of the 109 instances.  The cause is the ceiling of
$\lfloor n/3\rfloor$ on the size of a vertex-disjoint packing, which 36 of the
computed packings attain exactly: since $k>n/3$ on 103 of the instances, an
excess of at least $k-\lfloor n/3\rfloor$ is unavoidable there no matter how the
packing is chosen.

\end{document}